\documentclass[11pt, oneside]{article}   	
\usepackage{geometry}                		
\usepackage{graphicx}				

\usepackage{import}
\usepackage{authblk}						
\usepackage{cite}
\usepackage{array}								
\usepackage{amssymb}
\usepackage{amsmath}
\usepackage{mathtools}
\usepackage{braket}
\usepackage{amsthm}
\usepackage{dsfont}							
\usepackage{verbatim}						
\usepackage{MnSymbol}						
\usepackage{mathrsfs}						
\usepackage{mdframed}
\usepackage{caption}
\usepackage[dvipsnames]{xcolor}

\usepackage{float}
\usepackage{placeins}

\RequirePackage{doi}
\usepackage{hyperref}
\hypersetup{
  colorlinks   = true,    
  urlcolor     = blue,    
  linkcolor    = blue,    
  citecolor    = green      
}

\DeclareCaptionType{Protocol}

\usepackage{latexdef}

\title{Proving security of BB84 under source correlations}
\author{Ashutosh Marwah\footnote{email: ashutosh.marwah@outlook.com} \ and Fr\'ed\'eric Dupuis}
\affil{\small{D\'epartement d'informatique et de recherche op\'erationnelle,\\ Universit\'e de Montr\'eal,\\ Montr\'eal QC, Canada}}

\date{\today}							

\begin{document}
\maketitle

\begin{abstract}
    Device imperfections and memory effects can result in undesired correlations among the states generated by a realistic quantum source. These correlations are called source correlations. Proving the security of quantum key distribution (QKD) protocols in the presence of these correlations has been a persistent challenge. We present a simple and general method to reduce the security proof of the BB84 protocol with source correlations to one with an almost perfect source, for which security can be proven using previously known techniques. For this purpose, we introduce a simple source test, which randomly tests the output of the QKD source and provides a bound on the source correlations. We then use the recently proven entropic triangle inequality for the smooth min-entropy~\cite{Marwah23} to carry out the reduction to the protocol with the almost perfect source. 
\end{abstract}

\section{Introduction}

Quantum information enables the development of cryptographic protocols that surpass classical counterparts, offering not only enhanced security but also capabilities which would be impossible classically \cite{Bennett84,Konig12,Broadbent20, Bartusek23}. Protocols in quantum cryptography often require an honest party to produce multiple independent quantum states. As an example, quantum key distribution (QKD) protocols \cite{Bennett84,Bennet92} and bit commitment protocols \cite{Konig12,Lutkenhaus20} all require the honest participant, Alice to produce an independently and randomly chosen quantum state from a set of states in every round of the protocol. The security proofs for these protocols also rely on the fact that the quantum state produced in each round of the protocol is independent of the other rounds. However, this is a difficult property to enforce practically. All physical devices have an internal memory, which is difficult to characterise and control. This memory can cause the quantum states produced in different rounds to be correlated with one another. For example, when implementing BB84 states using the polarisation of light, if the polariser is in the horizontal polarisation ($\ket{0}$) for round $k$, and it is switched to the $\Pi/4$-diagonal polarisation ($\ket{+}$) in the $(k+1)$th round, then it is plausible that the state produced in the $(k+1)$th round is ``tilted'' towards the horizontal (that is, has a larger component along $\ket{0}$ than $\ket{1}$) simply due to the inertia of switching the polariser. Such correlations between different rounds caused by an imperfect source are called \emph{source correlations}. Security proofs for cryptographic protocols need to consider such correlations in order to be relevant in the real world. \\

An extensive line of research has led to techniques for proving the security of QKD protocols with a perfect source \cite{Shor00,Renner06,Koashi09,Tomamichel17,Dupuis20,Metger22-2}. However, almost all of these techniques rely on \emph{source purification}\footnote{Only \cite{Metger22-2} does not use source purification, but it still requires that the states in each round be produced independently.}-- the fact that the security of this protocol is equivalent to one where Alice sends out one half of a Bell state in each round and randomly measures her half. When the states produced by Alice's source are correlated across different rounds, this equivalence step fails and one can no longer use these methods. In this paper, we use the entropic triangle inequality recently proven in \cite[Lemma 3.5]{Marwah23} to reduce the security of the BB84-QKD protocol with source correlations to that of the BB84 protocol with a perfect source. With this reduction, one can simply use one of the many security analysis methods developed to complete the security proof\footnote{
    Following the reduction, any security proof technique for QKD which can bound $\tilde{H}^{\uparrow}_{\alpha}$ of Alice's raw key given Eve's side information can be used to complete the proof. The assumptions for the security of the protocol will be a combination of the assumptions required for this security proof and the assumptions used during the source test presented in Protocol \ref{frame:src_corr_prot}.
}. We demonstrate our technique using the BB84 protocol, although we believe it is quite general and can be applied to other cryptographic protocols as well. \\

\begin{figure}
    \centering
    \includegraphics[scale=0.1]{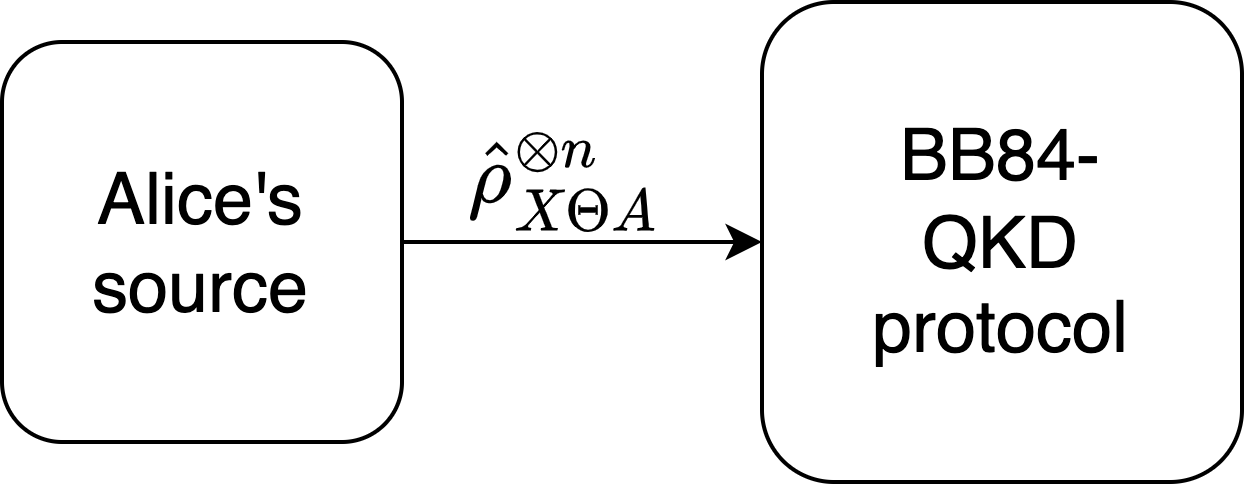}
    \caption{Quantum input for the BB84 protocol with a perfect source.}
    \label{fig:BB84-ideal}
\end{figure}

In the BB84 protocol, the only quantum state to the protocol is provided by Alice. The protocol can be represented as in Fig. \ref{fig:BB84-ideal}. If the source is imperfect and the BB84 protocol is directly performed on the state produced by such a source like in Fig. \ref{fig:BB84-ideal}, it is difficult to analyse the protocol and provide good security guarantees\footnote{
    The following example demonstrates the difficulty. Imagine a source which at the start of the QKD protocol flips a coin $C$, which is $0$ with probability $\epsilon_s$. If $C=1$, the source produces the qubit states perfectly, otherwise if $C=0$ it encodes 0 whenever a key generation basis is used. The state produced by this source will be $O(\epsilon_s)$ close to the perfect state in each round. It will also not abort during parameter estimation. However, with probability $\epsilon_s$ no key is produced between Alice and Bob. In this situation, we would like the protocol's secrecy error to remain arbitrarily small and its abort probability to be $\approx \epsilon_s$. For this we need to be able to somehow identify the $C=0$ bad case.
}. Instead, we propose and analyse the setup presented in Fig. \ref{fig:BB84-src-corr-prot}. Here the source is tested during the execution of the protocol using a simple procedure and the protocol run is only declared valid if the test passes. The source test randomly measures the output of the source on a small sample of the rounds in the preparation basis and aborts if the relative deviation of the observed output from the expected output is more than some small threshold $\epsilon$. Practically, this test can be carried out concurrently with the BB84 protocol and no quantum memory is required.\\

\begin{figure}
    \centering
    \includegraphics[scale=0.1]{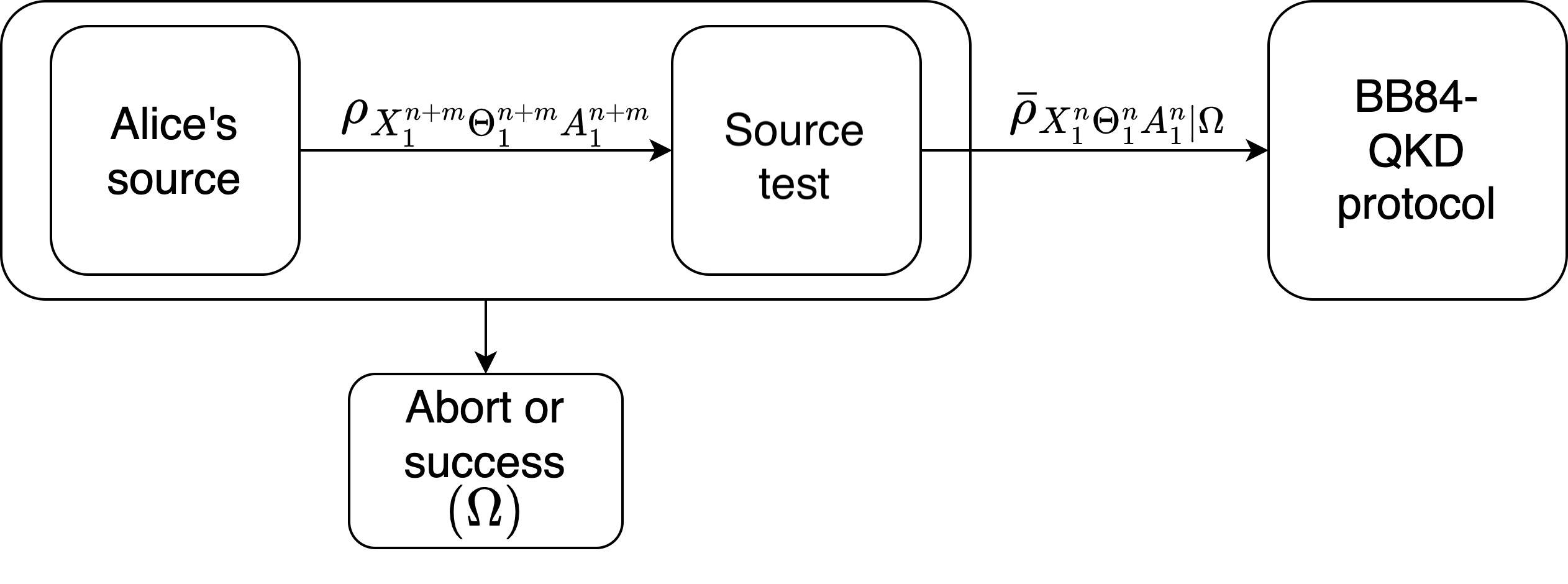}
    \caption{The setup for performing BB84 protocol with a source test.}
    \label{fig:BB84-src-corr-prot}
\end{figure}

For the security analysis of the protocol depicted in Fig. \ref{fig:BB84-src-corr-prot}, we do not need any assumptions on the source in addition to the BB84 protocol. Assumptions are only required on the measurements used for the source test. In Section \ref{sec:sec_proof_src_corr}, we present the security analysis assuming perfect measurements and then in Section \ref{sec:imp_meas}, we demonstrate how this analysis can be modified to incorporate imperfect measurements. It is worth noting that these measurements are used at a much smaller rate than the source, so it should be easier to implement them almost perfectly than it is to do the same for the source. \\

In comparison, \cite{Pereira22}, which provides a modified BB84 protocol for QKD in the presence of source imperfections and correlations, makes multiple complex assumptions about Alice's source (also see \cite{Lorenzo23}). Among these, it assumes that the state produced by Alice in the $k$th round can only be correlated to the states produced in the $\ell_c$ rounds preceding it, where $\ell_c$ is some known constant. Moreover, it also assumes that Alice's quantum states are not entangled across different rounds. These are both very strong assumptions, which cannot be guaranteed in practical setups. After our work appeared on Arxiv, a subsequent work \cite{Pereira24} improved the analysis in \cite{Pereira22}, enabling it to consider unbounded correlations. However, this analysis depends on a good model for the correlations and their decay. Given these \cite{Pereira24} shows that one can choose the parameter $\ell_c$ such that the protocol in \cite{Pereira22} leads to a secure QKD protocol with some additional error, which depends on how fast the source correlations decay\footnote{The additional security error for the protocol in \cite{Pereira24} is related to the sum of the correlations at distance $m \geq \ell_c$. Bounding this sum requires fast decay of correlations.}. Our work does not rely on any such model. In addition to reducing the assumptions required for QKD, this is also beneficial practically because the source parameters may change due to environmental factors, like temperature, as well as the age of the device. \\

The source test also addresses the challenge of characterising the source for QKD \cite{Mizutani19, Kang19, Huang23}. Most theoretical descriptions of QKD protocol require the source to operate almost perfectly. Thus, in order to implement these protocols one needs to characterise the source beforehand. Since we show security of BB84 as long as the source test succeeds, no prior characterisation is required for the source in our protocol. However, one still needs to characterise the measurements used in the source test. As mentioned above, this might be easier since the measurements are used at a much smaller rate. 

\section{Background and Notation}

For $n$ quantum registers $(X_1, X_2, \cdots, X_n)$, the notation $X_i^j$ refers to the set of registers $(X_i, X_{i+1}, \cdots, X_{j})$. We use the notation $[n]$ to denote the set $\{1,2, \cdots, n\}$. For a register $A$, $|A|$ represents the dimension of the underlying Hilbert space. If $X$ and $Y$ are Hermitian operators, then the operator inequality $X \geq Y$ denotes the fact that $X-Y$ is a positive semidefinite operator and $X>Y$ denotes that $X-Y$ is a strictly positive operator. A quantum state refers to a positive semidefinite operator with unit trace. At times, we will also need to consider positive semidefinite operators with trace less than equal to $1$. We call these operators subnormalised states. We will denote the set of registers a quantum state describes (equivalently, its Hilbert space) using a subscript. For example, a quantum state on the register $A$ and $B$, will be written as $\rho_{AB}$ and its partial states on registers $A$ and $B$, will be denoted as $\rho_{A}$ and $\rho_{B}$. The identity operator on register $A$ is denoted using $\Id_A$. A classical-quantum state on registers $X$ and $B$ is given by $\rho_{XB} = \sum_{x} p(x) \ket{x}\bra{x} \otimes \rho_{B|x}$, where $\rho_{B|x}$ are normalised quantum states on register $B$.\\

The term ``channel'' is used for completely positive trace preserving (CPTP) linear maps between two spaces of Hermitian operators. A channel $\cN$ mapping registers $A$ to $B$ will be denoted by $\cN_{A \rightarrow B}$. \\

The trace norm is defined as $\norm{X}_1 := \tr\big(\rndBrk{X^\dag X}^{\frac{1}{2}}\big)$. The fidelity between two positive operators $P$ and $Q$ is defined as $F(P,Q)= \norm{\sqrt{P}\sqrt{Q}}_1^2$. The generalised fidelity between two subnormalised states $\rho$ and $\sigma$ is defined as 
\begin{align}
    F_\ast(\rho, \sigma) := \rndBrk{\norm{\sqrt{\rho}\sqrt{\sigma}}_1 + \sqrt{(1- \tr\rho)(1- \tr\sigma)}}^2.
\end{align}
The purified distance between two subnormalised states $\rho$ and $\sigma$ is defined as 
\begin{align}
    P(\rho, \sigma) = \sqrt{1- F_{\ast}(\rho, \sigma)}.
\end{align}
Throughout this paper, we use base $2$ for both the functions $\log$ and $\exp$. We follow the notation in Tomamichel's book~\cite{TomamichelBook16} for R\'enyi entropies. The sandwiched $\alpha$-R\'enyi relative entropy for $\alpha \in [\frac{1}{2}, 1) \cup (1,\infty]$ between the positive operator $P$ and $Q$ is defined as 
\begin{align}
    \tilde{D}_{\alpha}(P || Q) = \begin{cases}
        \frac{1}{\alpha-1} \log \frac{\tr(Q^{-\frac{\alpha'}{2}} P Q^{-\frac{\alpha'}{2}})^{\alpha}}{\tr(P)} & \text{ if } (\alpha<1 \text{ and } P \not\perp Q) \text{ or } (P \ll Q)\\
        \infty & \text{ else}.
    \end{cases}
\end{align}
where $\alpha' = \frac{\alpha-1}{\alpha}$. In the limit $\alpha \rightarrow 
\infty$, the sandwiched divergence becomes equal to the max-relative entropy, $D_{\max}$, which is defined as 
\begin{align}
    D_{\max} (P|| Q) := \inf \curlyBrk{\lambda \in \mathbb{R}: P \leq 2^{\lambda}Q}.
\end{align}
In the limit of $\alpha \rightarrow 1$, the sandwiched relative entropy equals the quantum relative entropy, $D(P ||Q)$, which is defined as 
\begin{align}
    D(P|| Q) := \begin{cases}
        \frac{\tr\rndBrk{P \log P - P \log Q}}{\tr(P)} & \text{ if } P \ll Q\\
        \infty & \text{ else}.
    \end{cases}
\end{align}
We can use the sandwiched divergence to define the following conditional entropies for the subnormalised state $\rho_{AB}$:
\begin{align}
    \tilde{H}_{\alpha}^{\uparrow} (A|B)_{\rho} &:= \sup_{\sigma_B} - \tilde{D}_{\alpha}(\rho_{AB} || \Id_A \otimes \sigma_B) \\
    \tilde{H}_{\alpha}^{\downarrow} (A|B)_{\rho} &:= - \tilde{D}_{\alpha}(\rho_{AB} || \Id_A \otimes \rho_B)
\end{align}
again for $\alpha \in [\frac{1}{2}, 1) \cup (1,\infty]$. The supremum in the definition for $ \tilde{H}_{\alpha}^{\uparrow}$ is over all quantum states $\sigma_B$ on register $B$. \\

For $\alpha \rightarrow 1$, both these conditional entropies are equal to the von Neumann conditional entropy $H(A|B)$. $\tilde{H}_{\infty}^{\uparrow} (A|B)_{\rho}$ is usually called the min-entropy. The min-entropy is usually denoted as $H_{\min}(A|B)_{\rho}$ and for a subnormalised state can also be defined as 
\begin{align}
    H_{\min}(A|B)_{\rho} &:= \sup \curlyBrk{\lambda \in \mathbb{R}: \text{ there exists state }\sigma_B \text{ such that } \rho_{AB} \leq 2^{-\lambda} \Id_A \otimes \sigma_B}.
\end{align}  
For the purpose of smoothing, define the $\epsilon$-ball around the subnormalised state $\rho$ as the set
\begin{align}
    B_{\epsilon}(\rho) = \{ \tilde{\rho} \geq 0 : P(\rho, \tilde{\rho}) \leq \epsilon \text{ and } \tr\tilde{\rho} \leq 1\}.
\end{align}
We define the smooth max-relative entropy as 
\begin{align}
    D_{\max}^{\epsilon}(\rho || \sigma) = \min_{\tilde{\rho} \in B_{\epsilon}(\rho)} D_{\max}(\tilde{\rho} || \sigma)
\end{align}
The smooth min-entropy of $\rho_{AB}$ is defined as 
\begin{align}
    H_{\min}^{\epsilon}(A|B)_{\rho} = \max_{\tilde{\rho} \in B_{\epsilon}(\rho)} H_{\min}(A|B)_{\tilde{\rho}}.
\end{align} 

\section{Key lemmas}

In this section, we describe the two key results our analysis relies on. As described in the introduction, we use the entropic triangle inequality proven in \cite[Lemma 3.5]{Marwah23} to transform the problem of lower bounding $H^{\epsilon}_{\min}$ of the QKD protocol performed using the real correlated source to the problem of lower bounding $\tilde{H}^{\uparrow}_{\alpha}$ of a QKD protocol using perfect states. This inequality is stated in the following Lemma.  
\begin{lemma}[Entropic triangle inequality {\cite[Lemma 3.5]{Marwah23}}]
  For $\alpha \in (1,2]$, $\epsilon \in [0,1)$, and $\delta \in (0,1)$ such that $\epsilon + \delta < 1$ and two states $\rho$ and $\eta$, we have 
  \begin{align}
      H_{\min}^{\epsilon+\delta}(A|B)_{\rho} &\geq \tilde{H}^{\uparrow}_{\alpha}(A|B)_{\eta} - \frac{\alpha}{\alpha-1} D^\epsilon_{\max} (\rho_{AB}|| \eta_{AB}) - \frac{g_1(\delta, \epsilon)}{\alpha-1}
      \label{eq:Hmin_rho_to_Halpha_sigma_using_Dmax}
  \end{align}
  where $g_1(x, y):= - \log(1- \sqrt{1-x^2}) - \log (1-y^2)$. 
  \label{lemm:Hmin_rho_to_Halpha_sigma_using_Dmax}
\end{lemma}
In order to use the triangle inequality above effectively, we also need a way to bound the $D^{\epsilon}_{\max}$ between the output of the source test and the almost perfect source state. We will use results from \cite{Bouman10} for this task. \cite{Bouman10} studies how sampling techniques can be used to estimate the relative ``weight'' of a string classically as well as quantumly. In particular, it essentially generalises the Hoeffding-Serfling random sampling bounds to quantum states. The main result of this paper has been summarised as the Theorem below. To state it, we first need to define the relative weight of a string. For an alphabet $\mathcal{X}$ and a string $x_1^n \in \mathcal{X}^n$, the relative weight of $x_1^n$ is the frequency of non-zero $x_i$, that is, 
\begin{align}
    \omega(x_1^n) := \frac{1}{n}|\{ i \in [n]: x_i \neq 0 \}|.
\end{align}
Further, for a string $x_1^n$ and a subset $S \subseteq [n]$, let $x_S$ refer to the string $(x_i)_{i \in S}$.

\begin{theorem}[Quantum sampling \cite{Bouman10}]
    \label{th:qu_sampling}
    Let $\Psi$ be a sampling strategy which takes a string $a_1^n$, selects a random subset $\Gamma \subseteq [n]$ with probability $p_{\Gamma}$, a random seed $K$ with probability $p_K$ and produces an estimate $f(\Gamma, a_{\Gamma}, K)$ for the relative weight of the rest of the string $a_{\bar{\Gamma}}$. We can define the set of strings for which this strategy provides a $\delta$-correct estimate for $\delta>0$ given the choices $\Gamma = \gamma$ and $K= \kappa$ as 
    \begin{align}
        B^{\delta}_{\gamma \kappa}(\Psi) := \curlyBrk{a_1^n : |\omega(a_{\bar{\gamma}}) - f(\gamma, a_\gamma, \kappa)| < \delta}
    \end{align}
    where $\bar{\gamma}$ is the complement of the set $\gamma$ in $[n]$. The classical maximum error probability for this strategy $\Psi$ is defined as
    \begin{align}
        \epsilon^{\delta}_{\text{cl}} := \max_{a_1^n} \Pr_{\Gamma K}[a_1^n \not\in B^{\delta}_{\Gamma K}(\Psi)].
    \end{align}
    Define the projectors $\Pi_{A_1^n}^{\delta| \gamma \kappa}:= \sum_{a_1^n \in B^{\delta}_{\gamma \kappa}(\Psi)} \ket{a_1^n}\bra{a_1^n}_{A_1^n}$. Then, for a quantum state $\rho_{A_1^n E}$, we have that the state 
    \begin{align}
        \tilde{\rho}_{\Gamma K A_1^n E} := \sum_{\gamma \kappa} p(\gamma \kappa) \ket{\gamma \kappa}\bra{\gamma \kappa}_{\Gamma K} \otimes \frac{\Pi_{A_1^n}^{\delta| \gamma \kappa} \rho_{A_1^n E} \Pi_{A_1^n}^{\delta| \gamma \kappa}}{\tr\rndBrk{\Pi_{A_1^n}^{\delta| \gamma \kappa} \rho_{A_1^n E}}}
    \end{align}
    is $\epsilon^{\delta}_{\text{qu}} = \sqrt{\epsilon^{\delta}_{\text{cl}}}$ close to the state $\rho_{\Gamma K A_1^n E} := \sum_{\gamma \kappa} p(\gamma \kappa) \ket{\gamma \kappa}\bra{\gamma \kappa}_{\Gamma K} \otimes \rho_{A_1^n E}$ in trace distance. 
\end{theorem}
If one were to measure the string in the register given by $A_\gamma$ of the state $\tilde{\rho}$ defined above, then the rest of the registers $A_{\bar{\gamma}}$ of $\tilde{\rho}$ would lie in a subspace, which has relative weight $\delta$-close to $f(\gamma, a_\gamma, \kappa)$.

\section{Security proof for BB84 with source correlations}
\label{sec:sec_proof_src_corr}

We consider the BB84-QKD protocol described in Protocol \ref{frame:BB84}. In Table \ref{tab:var_defn}, we list all the variables we use for our proof along with their definitions.
\begin{figure}
    \begin{mdframed}
    \textbf{BB84 QKD protocol}\\

    \textbf{Parameters:}
    \begin{itemize}
        \item $n$ is the number of qubits sent by Alice.
        \item $\mu \in (0,1)$ is the probability of both encoding and measuring in the $X$ basis. 
        \item $e \in (0, \frac{1}{2})$ is the maximum error tolerated.
        \item $r \in (0,1)$ is the key rate of the protocol. 
    \end{itemize}

    \textbf{Protocol:}
    \begin{enumerate}
        \item For every $1 \leq i \leq n$ perform the following steps:
        \begin{enumerate}
            \item Alice chooses a random bit $X_i \in_R \{0,1\}$ and with probability $1-\mu$ encodes it in the $Z$ basis (as $\ket{0}$ or $\ket{1}$) and with probability $\mu$ in the $X$ basis (as $\ket{+}$ or $\ket{-}$). 
            \item Alice sends her encoded qubit to Bob. 
            \item Bob measures the qubit in the $Z$ basis with probability $1- \mu$ and in the $X$ basis with probability $\mu$. He records the output as $Y_i$.
        \end{enumerate}
        \item \textbf{Sifting:} Alice and Bob share their choice of bases for all the rounds and discard the rounds where their choices are different. We denote the remaining rounds by the set $S$. 
        \item \textbf{Information reconciliation:} Alice and Bob use an information reconciliation procedure, which lets Bob obtain a guess $\hat{X}_S$ for Alice's raw key $X_S$. 
        \item \textbf{Raw key validation:} Alice selects a random hash function from a two universal family and sends it along with $\text{hash}(X_S)$ to Bob. If $\text{hash}(X_S) \neq \text{hash}(\hat{X}_S)$ Bob aborts the protocol. 
        \item \textbf{Parameter estimation:} Let $S_X$ be the set of rounds where Alice prepared the qubit in the $X$ basis and Bob measured the qubit in $X$ basis. Bob aborts if ${\vert\{ i \in S_X : \hat{X}_i \neq Y_i\}\vert} > e \mu^2 n$.
        \item \textbf{Privacy Amplification:} Alice chooses a random function $F$ from a set of two universal hash functions from $|S|$ bits to $\lfloor rn \rfloor$ bits and announces it Bob. Alice and Bob compute the final key as $F(X_S)$ and $F(\hat{X}_S)$ respectively. 
    \end{enumerate}
    \end{mdframed}
    {\captionof{Protocol}{}
    \label{frame:BB84}}
\end{figure}
\begin{center}
    \begin{table}
    \begin{tabular}{ |l|p{0.8\textwidth}| } 
        \hline
        Variable & Definition  \\
        \hline \hline
        $\mathcal{X}$ & The set $\{0,1\}$; alphabet for Alice's random string.\\ \hline
        ${\Theta}$ & The set $\{0,1\}$; alphabet for the basis string.\\ \hline
        $X_1^n$  & The random string chosen by Alice at the beginning of the protocol. \\ \hline
        $\Theta_1^n$ & Alice's choice of randomly chosen basis. $\Theta_i =0 $ if Alice chooses $Z$ basis and $\Theta_i =1$ if she chooses $X$ basis. \\ \hline
        $A_1^n$ & The quantum registers sent by Alice to Bob. \\ \hline
        $\hat{\Theta}_1^n$ & Bob's choice of randomly chosen basis. $\hat{\Theta}_i =0 $ if Bob chooses $Z$ basis and $\hat{\Theta}_i =1$ if he chooses $X$ basis. \\ \hline
        $Y_1^n$ & Bob's outcomes of measuring $A_1^n$ in $\hat{\Theta}_1^n$ basis.\\ \hline
        $S$ & The set $\{i \in [n]: \hat{\Theta}_i = \Theta_i \}$. \\ \hline
        $\hat{X}_S$ & Bob's guess of $X_S$, produced at the end of the information reconciliation step.\\ \hline
        $T$ & Transcript for information reconciliation and raw key validation.\\ \hline
        $\bar{X}_1^n$ & For $i \in [n]$, $\bar{X}_i = X_i$ if $\Theta_i = \hat{\Theta}_i$ else $\bar{X}_i = \perp$. \\ \hline
        $\bar{Y}_1^n$ & For $i \in [n]$, $\bar{Y}_i = Y_i$ if $\Theta_i = \hat{\Theta}_i =1$ else $\bar{Y}_i = \perp$. \\ \hline
        $C_1^n$ & For $i \in [n]$, $C_i = X_i \oplus Y_i $ if $\Theta_i = \hat{\Theta}_i =1$ else $C_i = \perp$. \\ \hline
        $\hat{C}_1^n$ & For $i \in [n]$, $\hat{C}_i = \hat{X}_i \oplus Y_i $ if $\Theta_i = \hat{\Theta}_i =1$ else $\hat{C}_i = \perp$. \\ \hline
        $E$ & Eve's register created after Eve processes and forwards the states $A_1^n$ to Bob. \\ \hline
        $\Upsilon$ & The event that the protocol does not abort, i.e., ${\vert\{ i \in S_X : \hat{C}_i =1 \}\vert} \leq e \mu^2 n$ and $\text{hash}(X_S) = \text{hash}(\hat{X}_S)$.\\ \hline
        $\Upsilon'$ & The event that $X_S = \hat{X}_S$. \\ \hline
        $\Upsilon''$ & The event that ${\vert\{ i \in S_X : C_i =1 \}\vert} \leq e \mu^2 n$.\\
        \hline
    \end{tabular}
    \captionof{table}{Definition of variables for QKD \label{tab:var_defn}}
    \end{table}
\end{center}
At the beginning of every round of the QKD protocol, Alice prepares the classical registers $X_i$ and $\Theta_i$, and the corresponding qubit in the register $A_i$. If Alice's quantum source were perfect, she would produce the following state during each round of the protocol
\begin{align}
    \hat{\rho}_{X \Theta A} := \sum_{x \in \mathcal{X}, \theta \in \Theta} p(x, \theta) \ket{x, \theta}\bra{x, \theta}_{X \Theta} \otimes H^{\theta} \ket{x}\bra{x}_{A} H^{\theta}
\end{align}
where $H$ is the Hadamard gate and 
\begin{align}
    p(x, \theta) = \begin{cases}
        \frac{1- \mu}{|\mathcal{X}|} & \text{if } \theta = 0 \\
        \frac{\mu}{|\mathcal{X}|} & \text{if } \theta = 1.
    \end{cases}
\end{align}
Consider the case, where Alice only has access to an imperfect quantum source to prepare qubits for the QKD protocol above. We will assume here that the classical randomness used by Alice is perfect. We do not place any assumptions on the performance of the source. Suppose Alice and Bob use $n$ rounds for the BB84 protocol. In order, to perform the QKD protocol with the imperfect source, we require that Alice uses her source to first perform the source test given in Protocol \ref{frame:src_corr_prot} with $(n+m)$-total rounds. This test randomly selects $m$ rounds of the source output, measures the qubit $A_i$ in the basis given by $\Theta_i$ and compares the result with the encoded bit $X_i$ for these rounds. The source passes the test if the fraction of errors is less than $\epsilon$, which is a source error threshold chosen by Alice. Subsequently, Alice uses the $n$ remaining rounds produced by the source for the BB84 protocol provided the source test does not abort. The complete protocol is depicted in Fig. \ref{fig:BB84-src-corr-prot}. It should be noted that Alice can actually run Protocol \ref{frame:src_corr_prot} concurrently with the BB84 protocol. She does not need to create all the $(n+m)$-rounds at once and store them in a memory in order to carry out this protocol. She can classically sample a random set $\Gamma$ of size $m$ at the start of the BB84 protocol and for every round $i$, she can use the round as source test round if $i \in \Gamma$ or forward the state produced to Bob if $i \not\in \Gamma$. For theoretical purposes, this concurrent approach is equivalent to one where Alice begins by using her source to produce all the $(n+m)$-rounds and for our arguments we assume this is the case. In this section, we assume that the measurements used in the source test are perfect; we will then lift this assumption in Section~\ref{sec:imp_meas}. \\

\begin{figure}
    \begin{mdframed}
    \textbf{Source test}\\

    \textbf{Parameters:}
    \begin{itemize}
        \item $\epsilon$ is the source error tolerated.
        \item $m$ is the number of rounds on which the source is tested.
        \item $n$ is the number of rounds produced by the source for use in the subsequent protocol.
    \end{itemize}

    \textbf{Protocol:}
    \begin{enumerate}
        \item The source produces the state $\rho_{X_1^{n+m} \Theta_1^{n+m} A_1^{n+m}}$. 
        \item Choose a random subset $\Gamma \subseteq [n+m]$ of size $m$, measure the quantum registers $A_i$ in the basis given by $\Theta_i$ and let the result be $\hat{X}_i$. 
        \item Abort the protocol (and any subsequent protocols) if the observed error ${\frac{1}{m} |\{k \in \Gamma: \hat{X}_k \neq X_k \}| > \epsilon}$.
        \item Relabel the remaining registers from $1$ to $n$ and use them as the $n$ registers for the subsequent protocol.  
    \end{enumerate}
    \end{mdframed}
    {\captionof{Protocol}{}
    \label{frame:src_corr_prot}}
\end{figure}

Let $\rho_{X_1^{n+m} \Theta_1^{n+m} A_1^{n+m}}$ be the state produced by the imperfect source, $\Omega$ denote the event that the source test (Protocol \ref{frame:src_corr_prot}) does not abort and let the output of the source test protocol conditioned on $\Omega$ be the state $\bar{\rho}_{X_1^{n} \Theta_1^{n} A_1^{n}|\Omega}$ (or the subnormalised state $\bar{\rho}_{X_1^{n} \Theta_1^{n} A_1^{n}\wedge\Omega}$ depending on the context).\\

In the following Lemma, we prove that $\bar{\rho}_{X_1^{n} \Theta_1^{n} A_1^{n}|\Omega}$ has a relatively small smooth max-relative entropy with a depolarised version of the perfect source using Theorem \ref{th:qu_sampling}. We then use the entropic triangle inequality to reduce the security of the BB84 protocol with the imperfect source to that of a BB84 protocol, which uses this state as its source state.

\begin{lemma}
    \label{lemm:smoothDmax_bd_perf_meas}
    Let $\epsilon$ be the threshold of the source test, $\delta \in (0,1)$ a small parameter, and let $\hat{\rho}_{X \Theta A}^{(\epsilon +\delta)} := (1 - 2(\epsilon + \delta)) \hat{\rho}_{X \Theta A} + 2(\epsilon + \delta) \hat{\rho}_{X \Theta}\otimes \tau_A$ where $\tau_A$ is the completely mixed state on the register $A$. For the state $\bar{\rho}_{X_1^{n} \Theta_1^{n} A_1^{n}|\Omega}$ produced by the source test conditioned on passing, we have that
    \begin{align}
        D_{\max}^{\epsilon_f} (\bar{\rho}_{X_1^{n} \Theta_1^{n} A_1^{n} | \Omega} || \rndBrk{\hat{\rho}_{X \Theta A}^{(\epsilon +\delta)}}^{\otimes n})  \leq n h(\epsilon + \delta) + \log \frac{1}{\Pr_{{\rho}} (\Omega) - \epsilon_{\text{qu}}^{\delta}}
        \label{eq:src_corr_smooth_Dmax_bd}
    \end{align}
    where $\Pr_{\rho}(\Omega)$ is the probability of the event $\Omega$ when the testing procedure is applied to the state $\rho$, $h(x) = - x \log(x) - (1-x)\log (1-x)$ is the binary entropy function and $\epsilon_f = 2 \sqrt{\frac{\epsilon_{\text{qu}}^{\delta}}{\Pr_\rho (\Omega)}}$ for $\epsilon_{\text{qu}}^{\delta} = \sqrt{2}\exp\rndBrk{- \frac{n \delta^2}{2(n+2)} m }$.
\end{lemma}
\begin{proof}
    Define the unitaries, 
    \begin{align} 
        &V^{x, \theta}_A := H^{\theta} X^{x} \\
        &V_{X \Theta A} := \sum_{x, \theta} \ket{x, \theta} \bra{x, \theta}_{X \Theta} \otimes V^{x, \theta}_A 
        \label{eq:V_unit_defn}
    \end{align}
    so that $V_{X \Theta A} \ket{x, \theta}\ket{0}$ gives the perfect encoding of the BB84 state given $x$ and $\theta$. We also define the state 
    \begin{align}
        \nu_{X_1^{n+m} \Theta_1^{n+m} A_1^{n+m}} := \bigotimes_{i=1}^{n+m} V^\dagger_{X_i \Theta_i A_i} \ \rho_{X_1^{n+m} \Theta_1^{n+m} A_1^{n+m}}\ \bigotimes_{i=1}^{n+m} V_{X_i \Theta_i A_i}.
        \label{eq:defn_nu}
    \end{align}
    Note that if $\rho$ were perfectly encoded, then $\nu$ would be the state $\rho_{X_1^{n+m} \Theta_1^{n+m}} \otimes \rndBrk{\ket{0}\bra{0}}^{\otimes (n+m)}$. Let the register $\Gamma$ represent the choice of the random subset for sampling following the notation in Theorem \ref{th:qu_sampling}. The state produced by measuring the subset $\gamma$ of the $A$ registers of $\nu$ in the computational ($\{\ket{0}, \ket{1}\}$) basis can equivalently be produced by measuring the subset $\gamma$ of the $A$ registers of $\rho$ in the basis given by the corresponding $\Theta$ registers, adding (mod 2) the corresponding $X$ register to the result and applying the unitaries $V_{X \Theta A}$ on the remaining indices. Conditioning on the sampled qubits of $\rho$ being incorrectly encoded at most an $\epsilon$ fraction of the rounds is equivalent to measuring the corresponding random subset of the qubits of $\nu$ in the computational basis and conditioning on the relative weight of the result being less than $\epsilon$ (up to unitaries on the remaining registers; formal expression is given in Eq. \ref{eq:rho_nu_cond_eq}). Given this equivalence, we can simply work with the state $\nu$ and transform the results back to the state $\rho$ at the end. \\

    Using Theorem \ref{th:qu_sampling}, we have that for every $x_1^{n+m}, \theta_1^{n+m}$ there exists ${\eta}_{\Gamma A_1^{n+m} | x_1^{n+m}, \theta_1^{n+m}}$ such that 
    \begin{align}
        \frac{1}{2}\norm{\nu_{\Gamma A_1^{n+m} | x_1^{n+m}, \theta_1^{n+m}} - {\eta}_{\Gamma A_1^{n+m} | x_1^{n+m}, \theta_1^{n+m}}}_1 \leq \epsilon_{\text{qu}}^{\delta}
    \end{align}
    and  
    \begin{align}
        {\eta}_{\Gamma A_1^{n+m} | x_1^{n+m}, \theta_1^{n+m}} := \sum_{\gamma} p(\gamma) \ket{\gamma}\bra{\gamma} \otimes {\eta}^{(\gamma)}_{A_1^{n+m} | x_1^{n+m}, \theta_1^{n+m}}
    \end{align}
    where $p(\gamma)$ is the uniform distribution over all size $m$ subsets of $[n+m]$, and the state ${\eta}^{(\gamma)}_{A_1^{n+m} | x_1^{n+m}, \theta_1^{n+m}}$ satisfies
    \begin{align}
        {\eta}^{(\gamma)}_{A_1^{n+m} | x_1^{n+m}, \theta_1^{n+m}} = \Pi_{A_1^{n+m}}^{\delta|\gamma} {\eta}^{(\gamma)}_{A_1^{n+m} | x_1^{n+m}, \theta_1^{n+m}} \Pi_{A_1^{n+m}}^{\delta|\gamma}
        \label{eq:cond_eta_lies_in_projector}
    \end{align}
    for the projectors $\Pi_{A_1^{n+m}}^{\delta|\gamma}$ defined as in Theorem \ref{th:qu_sampling} (our sampling procedure does not require a random seed $\kappa$, so we omit it in our analysis). Note that using Hoeffding's bound the classical error probability for our sampling strategy is $2\exp(- \frac{n \delta^2}{n+2} m)$, which implies that $\epsilon_{\text{qu}}^{\delta} = \sqrt{2}\exp(- \frac{n \delta^2}{2(n+2)} m )$. We can also define the extended state ${\eta}_{\Gamma X_1^{n+m} \Theta_1^{n+m} A_1^{n+m}}$ as 
    \begin{align}
        {\eta}_{\Gamma X_1^{n+m} \Theta_1^{n+m} A_1^{n+m}} &:= \sum_{\gamma, x_1^{n+m}, \theta_1^{n+m}} p (\gamma) p(x_1^{n+m}, \theta_1^{n+m}) \nonumber\\
        & \qquad \qquad \ket{\gamma, x_1^{n+m}, \theta_1^{n+m}} \bra{\gamma, x_1^{n+m}, \theta_1^{n+m}} \otimes {\eta}^{(\gamma)}_{A_1^{n+m} | x_1^{n+m}, \theta_1^{n+m}}
        \label{eq:eta_defn}
    \end{align}
    where $p(x_1^{n+m}, \theta_1^{n+m}) = \prod_{i=1}^{n+m} p(x_i, \theta_i)$. Since, $\nu_{\Gamma X_1^{n+m} \Theta_1^{n+m} A_1^{n+m}}$ and ${\eta}_{\Gamma X_1^{n+m} \Theta_1^{n+m} A_1^{n+m}}$ have the same distributions on $X_1^{n+m}$ and $\Theta_1^{n+m}$, we also have that 
    \begin{align}
        \frac{1}{2}\norm{\nu_{\Gamma X_1^{n+m} \Theta_1^{n+m} A_1^{n+m}} - {\eta}_{\Gamma X_1^{n+m} \Theta_1^{n+m} A_1^{n+m}}}_1 \leq \epsilon_{\text{qu}}^{\delta}.
        \label{eq:nu_eta_dist_bd}
    \end{align}
    Define $\Omega'$ to be the event that the result produced by measuring the subset of registers $A_{\gamma}$ in the computational basis, where $\gamma$ is given by the $\Gamma$ register, has a relative weight less than $\epsilon$. Let $\bar{\nu}_{\Gamma X_1^{n} \Theta_1^{n} A_1^{n} \wedge \Omega'}$ be the subnormalised state produced when the relative weight of the registers $A_{\gamma}$ of $\nu_{\Gamma X_1^{n+m} \Theta_1^{n+m} A_1^{n+m}}$ is measured and conditioned on $\Omega'$, the registers $X_\gamma$ and $\Theta_\gamma$ are traced over, and the remaining $X, \Theta$ and $A$ registers are relabelled between $1$ and $n$. Also, let $\bar{\eta}_{\Gamma X_1^{n} \Theta_1^{n} A_1^{n} \wedge \Omega'}$ be the subnormalised state produced when this same subnormalised channel is instead applied to ${\eta}_{\Gamma X_1^{n+m} \Theta_1^{n+m} A_1^{n+m}}$. Let us consider the action of this map on a general state $\ket{\gamma}\bra{\gamma}\otimes \sigma_{A_1^{n+m}}^{(\gamma)}$, which satisfies the condition $\sigma_{A_1^{n+m}}^{(\gamma)} = \Pi_{A_1^{n+m}}^{\delta| \gamma} \sigma_{A_1^{n+m}}^{(\gamma)} \Pi_{A_1^{n+m}}^{\delta| \gamma}$. We can write such a state in the following form:
    \begin{align*}
        \sigma_{A_1^{n+m}}^{(\gamma)} &= \sum_{a_1^{n+m}, \bar{a}_1^{n+m} \in B^{\delta}_\gamma} \sigma^{(\gamma)}(a_1^{n+m}, \bar{a}_1^{n+m}) \ket{a_1^{n+m}}\bra{\bar{a}_1^{n+m}}
        \numberthis
    \end{align*}
    where $\sigma^{(\gamma)}(a_1^{n+m}, \bar{a}_1^{n+m}) := \bra{a_1^{n+m}}\sigma^{(\gamma)}\ket{\bar{a}_1^{n+m}}$ are the matrix elements for $\sigma^{(\gamma)}$. Let $\hat{P}_{A_1^m} := \sum_{a_1^m : \omega(a_1^m)\leq \epsilon} \ket{a_1^{m}} \bra{a_1^{m}}$ be the (perfect) measurement operator for conditioning on the event $\Omega'$. Then, the state after applying the measurement and conditioning on the $\Omega'$ is
    \begin{align*}
        \tr_{A_\gamma} \rndBrk{\hat{P}_{A_\gamma} \sigma_{A_1^{n+m}}^{(\gamma)}} &= \sum_{a_\gamma:  \omega(a_\gamma) \leq \epsilon} \sum_{\substack{a'_{\bar{\gamma}}, \bar{a}_{\bar{\gamma}} \in \{x_1^n :\\ 
        |\omega(x_1^n) - \omega(a_\gamma)| < \delta \} }} \sigma^{(\gamma)}(a_\gamma a'_{\bar{\gamma}}, a_\gamma \bar{a}_{\bar{\gamma}}) \ket{a'_{\bar{\gamma}}}\bra{\bar{a}_{\bar{\gamma}}}. \numberthis
    \end{align*}
    We can relabel the remaining registers to get the state $\bar{\sigma}_{A_1^{n} \wedge \Omega'}^{(\gamma)}$ which can be put into the form
    \begin{align}
        \bar{\sigma}_{A_1^{n} \wedge \Omega'}^{(\gamma)} &= \sum_{a_1^{n}, \bar{a}_1^{n}: \in \curlyBrk{x_1^n :\ \omega(x_1^n) < \epsilon + \delta} } \bar{\sigma}^{(\gamma)}(a_1^{n}, \bar{a}_1^{n}) \ket{a_1^{n}}\bra{\bar{a}_1^{n}}.
    \end{align}
    Let $Q^{(w)}_{A_1^n}$ be the projector on the set $\text{span}\{ \ket{x_1^n} :\ \text{for $x_1^n$ such that } \omega(x_1^n) < w \}$ (note that these vectors are perpendicular). Then, we have that 
    \begin{align}
        \bar{\sigma}_{A_1^{n} \wedge \Omega'}^{(\gamma)} = Q_{A_1^n}^{(\epsilon + \delta)} \bar{\sigma}_{A_1^{n} \wedge \Omega'}^{(\gamma)} Q_{A_1^n}^{(\epsilon + \delta)}
    \end{align}
    which implies that $\bar{\sigma}_{A_1^{n} \wedge \Omega'}^{(\gamma)} \leq Q_{A_1^n}^{(\epsilon + \delta)}$, since $\bar{\sigma}_{A_1^{n} \wedge \Omega'}^{(\gamma)}$ is subnormalised. \\

    By considering $\sigma_{A_1^{n+m}}^{(\gamma)} = \eta_{A_1^{n+m} | x_1^{n+m} \theta_1^{n+m}}^{(\gamma)}$, we see that $\bar{\eta}$ satisfies
    \begin{align*}
        \bar{\eta}_{\Gamma X_1^{n} \Theta_1^{n} A_1^{n} \wedge \Omega'} &= \sum_{\gamma, x_1^n, \theta_1^n} p(\gamma)p (x_1^n \theta_1^n) \ket{\gamma x_1^n \theta_1^n}\bra{\gamma x_1^n \theta_1^n} \otimes  \bar{\eta}^{(\gamma)}_{A_1^n | x_1^n \theta_1^n \wedge \Omega'} \\
        &\leq \sum_{\gamma, x_1^n, \theta_1^n} p(\gamma)p (x_1^n \theta_1^n) \ket{\gamma x_1^n \theta_1^n}\bra{\gamma x_1^n \theta_1^n} \otimes Q_{A_1^n}^{(\epsilon + \delta)}\\
        &= \rho_{\Gamma} \otimes \rho_{X \Theta}^{\otimes n} \otimes Q_{A_1^n}^{(\epsilon + \delta)}.
        \numberthis
    \end{align*}
    Using the data processing inequality, we also have that
    \begin{align}
        \frac{1}{2}\norm{\bar{\nu}_{\Gamma X_1^{n} \Theta_1^{n} A_1^{n} \wedge \Omega'} - \bar{\eta}_{\Gamma X_1^{n} \Theta_1^{n} A_1^{n} \wedge \Omega'}}_1 \leq \epsilon_{\text{qu}}^{\delta}
        \label{eq:conditioned_barNu_close_to_barEta}
    \end{align}
    Let $\hat{\eta}^{(\epsilon +\delta)}_A := (1- \epsilon - \delta)\ket{0}\bra{0} + (\epsilon + \delta)\ket{1}\bra{1}$ or equivalently the state $\hat{\eta}^{(\epsilon +\delta)}_A$ is the classical probability distribution over $\{0,1\}$ which is $1$ with probability $(\epsilon + \delta)$. For this distribution, a simple calculation shows that 
    \begin{align*}
        \min_{z_1^n: \omega(z_1^n) < \epsilon + \delta}\bra{z_1^n} (\hat{\eta}^{(\epsilon +\delta)}_A)^{\otimes n} \ket{z_1^n} \geq  2^{-n h(\epsilon + \delta)}
        \numberthis
    \end{align*}
    which implies that 
    \begin{align*}
        Q_{A_1^n}^{(\epsilon + \delta)} &\leq 2^{n h(\epsilon + \delta)} (\hat{\eta}^{(\epsilon +\delta)}_A)^{\otimes n}
        \numberthis
    \end{align*}
    since both $Q_{A_1^n}^{(\epsilon + \delta)}$ and $(\hat{\eta}^{(\epsilon +\delta)}_A)^{\otimes n}$ share the classical eigenbasis $\curlyBrk{\ket{z_1^n} : z_1^n \in \{0,1\}^n}$. Thus, we have 
    \begin{align}
        \bar{\eta}_{X_1^{n} \Theta_1^{n} A_1^{n} \wedge \Omega'} &\leq 2^{n h(\epsilon + \delta)} \rndBrk{\rho_{X \Theta} \otimes \hat{\eta}^{(\epsilon +\delta)}_A}^{\otimes n}.
    \end{align}
    As noted earlier, the state produced by measuring the registers $A_\gamma$ of $\nu$ in the computational basis is the same as the state produced by measuring the same registers on the real state $\rho$ in the basis given by $\Theta_i$, adding $X_i$ to the result (mod 2), and transforming the remaining registers with $\bigotimes_{k=1}^n V^{\dagger}_{X_i \Theta_i A_i}$. Under this correspondence, we have that the state produced by the source test satisfies
    \begin{align}
        \bar{\rho}_{X_1^{n} \Theta_1^{n} A_1^{n} \wedge \Omega} = \bigotimes_{i=1}^n V_{X_i \Theta_i A_i} \bar{\nu}_{X_1^{n} \Theta_1^{n} A_1^{n} \wedge \Omega'} \bigotimes_{i=1}^n V^\dagger_{X_i \Theta_i A_i}.
        \label{eq:rho_nu_cond_eq}
    \end{align}
    Further, for the state defined as
    \begin{align*}
        \bar{\tilde{\rho}}_{X_1^{n} \Theta_1^{n} A_1^{n} \wedge \Omega} &:= \bigotimes_{i=1}^n V_{X_i \Theta_i A_i} \bar{\eta}_{X_1^{n} \Theta_1^{n} A_1^{n} \wedge \Omega'} \bigotimes_{i=1}^n V^\dagger_{X_i \Theta_i A_i} \\
        &\leq 2^{n h(\epsilon + \delta)} \bigotimes_{i=1}^n \rndBrk{V_{X_i \Theta_i A_i}\ \rho_{X_i \Theta_i} \otimes \hat{\eta}^{(\epsilon +\delta)}_{A_i} V_{X_i \Theta_i A_i}^\dagger} \\
        &= 2^{n h(\epsilon + \delta)} \rndBrk{\hat{\rho}_{X \Theta A}^{(\epsilon +\delta)}}^{\otimes n} \numberthis 
        \label{eq:conditioned_rho_tilde_Dmax}
    \end{align*}
    where $\hat{\rho}_{X \Theta A}^{(\epsilon +\delta)} := (1 - 2(\epsilon + \delta)) \hat{\rho}_{X \Theta A} + 2(\epsilon + \delta) \hat{\rho}_{X \Theta}\otimes \tau_A$ for the completely mixed state $\tau_A$ on register $A$. Using Eq. \ref{eq:conditioned_barNu_close_to_barEta}, we also have
    \begin{align}
        \frac{1}{2}\norm{\bar{\rho}_{X_1^{n} \Theta_1^{n} A_1^{n} \wedge \Omega} - \bar{\tilde{\rho}}_{X_1^{n} \Theta_1^{n} A_1^{n} \wedge \Omega}}_1 \leq \epsilon_{\text{qu}}^{\delta}.
    \end{align}
    Following the argument in \cite[Lemma G.1]{Marwah23}, we can show that
    \begin{align}
        \frac{1}{2}\norm{\bar{\rho}_{X_1^{n} \Theta_1^{n} A_1^{n} | \Omega} - \bar{\tilde{\rho}}_{X_1^{n} \Theta_1^{n} A_1^{n} | \Omega}}_1 \leq \frac{2 \epsilon_{\text{qu}}^{\delta}}{\Pr_{{\rho}} (\Omega)}
        \label{eq:src_corr_dist_bd}
    \end{align}
    where $\Pr_{{\rho}} (\Omega) := \tr\rndBrk{\bar{\rho}_{X_1^{n} \Theta_1^{n} A_1^{n} \wedge \Omega}}$ is the probability of the event $\Omega$ when the testing procedure is applied to the state $\rho$, and 
    \begin{align*}
        \bar{\tilde{\rho}}_{X_1^{n} \Theta_1^{n} A_1^{n} | \Omega} &\leq \frac{2^{n h(\epsilon + \delta)}}{\Pr_{{\tilde{\rho}}} (\Omega)}\rndBrk{\hat{\rho}_{X \Theta A}^{(\epsilon +\delta)}}^{\otimes n} \\
        &\leq \frac{2^{n h(\epsilon + \delta)}}{\Pr_{{\rho}} (\Omega) - \epsilon_{\text{qu}}^{\delta}}\rndBrk{\hat{\rho}_{X \Theta A}^{(\epsilon +\delta)}}^{\otimes n} \numberthis
        \label{eq:src_corr_Dmax_bd}
    \end{align*}
    where $\Pr_{{\tilde{\rho}}} (\Omega) := \tr\rndBrk{\bar{\tilde{\rho}}_{X_1^{n} \Theta_1^{n} A_1^{n} \wedge \Omega}}$ is defined similar to $\Pr_{{\rho}} (\Omega)$. Together these imply that
    \begin{align}
        D_{\max}^{\epsilon_f} (\bar{\rho}_{X_1^{n} \Theta_1^{n} A_1^{n} | \Omega} || \rndBrk{\hat{\rho}_{X \Theta A}^{(\epsilon +\delta)}}^{\otimes n})  \leq n h(\epsilon + \delta) + \log \frac{1}{\Pr_{{\rho}} (\Omega) - \epsilon_{\text{qu}}^{\delta}}
        \label{eq:src_corr_smooth_Dmax_bd2}
    \end{align}
    where $\epsilon_f = 2 \sqrt{\frac{\epsilon_{\text{qu}}^{\delta}}{\Pr_\rho (\Omega)}}$.
\end{proof}

We now give an outline for bounding the smooth min-entropy for a BB84-QKD protocol, which uses an imperfect source. We give a complete formal proof in Appendix \ref{sec:formal_src_qkd_pf}. Let $\Phi_{\text{QKD}}$ be the CPTP map denoting the action of the entire QKD protocol on the source states produced by Alice. In order to prove security for QKD, informally speaking, it is sufficient to prove a linear lower bound for\footnote{We also need to condition on the QKD protocol not aborting. We do this in Appendix \ref{sec:formal_src_qkd_pf}}
\begin{align*}
    H^{\epsilon_f + \epsilon'}_{\min} (X_S | E T \Theta_1^n \hat{\Theta}_1^n)_{\Phi_{\text{QKD}} (\bar{\rho}_{|\Omega})}.
\end{align*}
Let us define the virtual state $\sigma_{X_1^n \Theta_1^n A_1^n} :=\rndBrk{\hat{\rho}_{X \Theta A}^{(\epsilon +\delta)}}^{\otimes n}$. This state can be viewed as the state produced when each of the qubits produced by Alice is passed through a depolarising channel. Using Lemma \ref{lemm:Hmin_rho_to_Halpha_sigma_using_Dmax}, for an arbitrary $\epsilon' >0$, we have 
\begin{align*}
    H^{\epsilon_f + \epsilon'}_{\min} &(X_S | E T \Theta_1^n \hat{\Theta}_1^n)_{\Phi_{\text{QKD}} (\bar{\rho}_{|\Omega})} \\
    &\geq \tilde{H}^{\uparrow}_{\alpha}(X_S | E T \Theta_1^n \hat{\Theta}_1^n)_{\Phi_{\text{QKD}} (\sigma)} - \frac{\alpha}{\alpha-1} D^{\epsilon_f}_{\max} (\Phi_{\text{QKD}} (\bar{\rho}_{|\Omega})|| \Phi_{\text{QKD}} (\sigma))- \frac{g_1(\epsilon', \epsilon_f)}{\alpha-1} \\
    &\geq \tilde{H}^{\uparrow}_{\alpha}(X_S | E T \Theta_1^n \hat{\Theta}_1^n)_{\Phi_{\text{QKD}} (\sigma)} - \frac{\alpha}{\alpha-1} n h(\epsilon+\delta) - \frac{O(1)}{\alpha-1}.
    \numberthis
\end{align*}
Thus, it is sufficient to bound the $\alpha$-R\'enyi conditional entropy $\tilde{H}^{\uparrow}_{\alpha}(X_S | E T \Theta_1^n \hat{\Theta}_1^n)$ for the QKD protocol running on a noisy version of the perfect source. We can now simply use standard techniques developed for the security proofs of QKD to show a linear lower bound for this conditional entropy. In particular, source purification can be used for the source state $\sigma$. In Appendix \ref{sec:formal_src_qkd_pf}, we show how one can modify the security proof for BB84-QKD based on entropy accumulation \cite[Section 5.1]{Dupuis20} to get the following bound.

\begin{theorem}
    Suppose Alice uses the output of the source test (Protocol \ref{frame:src_corr_prot}), with error threshold $\epsilon$ and any imperfect source as its input, as her source for the BB84 protocol. Let $\delta>0$ and assume that $h(\epsilon+\delta)< \frac{1}{\sqrt{2}}$. Then, for  
    \begin{align}
        & \epsilon^{\delta}_{\text{qu}} = \sqrt{2} \exp\rndBrk{-\frac{n\delta^2}{2(n+2)}m}\\
        & \epsilon_{\text{pa}} = 2\rndBrk{\frac{2\epsilon^{\delta}_{\text{qu}}}{\Pr_{\bar{\rho}}(\Omega \wedge \Upsilon'')}}^{1/2}
    \end{align}
    and $\epsilon'>0$, we have the following lower bound on the smooth min-entropy for the raw key produced during the BB84 protocol
    \begin{align*}
        H&^{\epsilon_{\text{pa}} + \epsilon'}_{\min} (X_S | E \Theta_1^n \hat{\Theta}_1^n T)_{\Phi_{\text{QKD}} (\bar{\rho})_{|\Omega \wedge \Upsilon''}} \\
        &\geq n(1- 2 \mu  - h(e) - V \sqrt{2h(\epsilon+\delta)}) - \sqrt{n} \rndBrk{\mu^2 \ln(2) + 2 \log\frac{1}{\Pr_{\bar{\rho}}(\Omega \wedge \Upsilon'')} +{g_0\rndBrk{\frac{\epsilon'}{8}}}}\\
        &- \frac{V}{\sqrt{2h(\epsilon+\delta)}} \rndBrk{\log \frac{1}{\Pr_{\bar{\rho}}(\Omega \wedge \Upsilon'') - 2\epsilon^{\delta}_{\text{qu}}} + 1} - \frac{g_1(\frac{\epsilon'}{2}, \epsilon_{\text{pa}})}{2\sqrt{2h(\epsilon+\delta)}} V   - \log|T| - 3 g_0\rndBrk{\frac{\epsilon'}{8}}
        \numberthis
    \end{align*}
    \label{th:src_cor_Hmin_bd}
    where $V := \frac{2}{\mu^2}\log\frac{1-e}{e} + 2 \log(1+2|\mathcal{X}|^2)$, $\Pr_{\bar{\rho}}(\Omega \wedge \Upsilon'')$ is the probability of the event $\Omega \wedge \Upsilon''$ for the state $\Phi_{\text{QKD}}(\bar{\rho})$ and it is assumed that $\Pr_{\bar{\rho}}(\Omega \wedge \Upsilon'') > 2 \epsilon^{\delta}_{\text{qu}}$\footnote{If $\Pr_{\bar{\rho}}(\Omega \wedge \Upsilon'') \leq 2 \epsilon^{\delta}_{\text{qu}}$, one can easily show that the secrecy condition for QKD is satisfied for a security parameter greater than $2 \epsilon^{\delta}_{\text{qu}}$ since this condition is weighted by the abort probability: $\Pr(\neg \text{Abort})\frac{1}{2}\norm{\rho^{(\text{real})}_{KE} - \rho^{(\text{perfect})}_{KE}}_1 \leq \epsilon_{sec}$ \cite[Section 4.3]{Portmann14}.}, $g_0(x) = - \log(1 - \sqrt{1-x^2})$ and $g_1(x,y) = - \log(1 - \sqrt{1-x^2}) - \log (1- y^2)$.
\end{theorem}

According to the Theorem above, the asymptotic key rate for the BB84 protocol using an imperfect source is $V\sqrt{2h(\epsilon + \delta)}$ lesser than a protocol, which uses a perfect source. Our aim was to provide a complete analytical security proof, so we have not optimised for constants and there remains significant room for improvement. That being said, let's analyse the behaviour of the rate loss $V\sqrt{2h(\epsilon + \delta)}$. In our analysis, the dependence of $V$ on $\mu$ is quite poor. We have $V = \Theta(1/{\mu^2})$ above. This can be improved to $\Theta(1/{\mu})$ fairly easily by using the entropy accumulation theorem with the improved second order term \cite{Dupuis19}. \cite{Dupuis19} also shows that for a BB84 protocol with perfect source, $\mu$ can be chosen as small as $c/\sqrt{n}$ for some constant $c>0$. Clearly, one cannot choose such small values ($o(1)$) for $\mu$ for our protocol. Further work is required to understand the dependence of the rate loss on $\mu$. It might be possible to avoid it using an entropic uncertainty relation based proof to bound $\tilde{H}^{\uparrow}_{\alpha}$ (similar to \cite{Tomamichel12} but using \cite[Theorem 1]{Coles12}). In the bound above, $V$ also increases as the error threshold $e$ is made smaller. This appears to be a result of the choice of linearisation used in the entropy accumulation theorem. This also makes the key rate corresponding to the bound above increasing in $e$ for small $e$. A simple solution to make the rate monotonically decreasing in the error would be to maximise over thresholds $e$ which are greater than the expected error\footnote{This was pointed to us by Ernest Tan.}. However, we leave the bound in the form above as it clearly shows the additional loss due to source correlations. 

\section{Imperfect measurements}
\label{sec:imp_meas}

In our analysis above, we assumed that the measurements used in the source test are perfect. It should be noted that if the source produces states at a rate $r_s$, then the measurement device is only used at an average rate $\frac{m}{n+m}r_s$, which is much smaller than $r_s$. So, the measurement devices have a much longer relaxation time than the source. As such, it should be easier to create almost ``perfect'' measurement devices than it is to create perfect sources. \\

In this section, we will show how measurement imperfections can also be incorporated in our analysis. Let $\Lambda\rndBrk{\leq \epsilon | \gamma, x_\gamma , \theta_\gamma}_{A_\gamma}$ be the POVM element associated with the source test passing, i.e., with measuring a relative weight less than $\epsilon$ with respect to the encoded random bits given the choice of random subset $\gamma$, encoded random bits $x_\gamma$, and basis choice $\theta_\gamma$. Informally speaking, in this subsection, we assume that this measurement measures the relative weight with an error at most $\epsilon_m$ with high probability. To formally state our assumption, define
\begin{align}
    & \hat{P}^{ x_\gamma, \theta_\gamma}_{A_\gamma} := \bigotimes_{i\in \gamma} V^{x_i, \theta_i}_{A_i} \rndBrk{\sum_{a_\gamma: \omega(a_\gamma) \leq \epsilon + \epsilon_m} \ket{a_\gamma}\bra{a_\gamma}_{A_\gamma}} \bigotimes_{i \in \gamma} \rndBrk{V^{x_i, \theta_i}_{A_i}}^\dagger \\
    & \hat{P}^{\perp| x_\gamma, \theta_\gamma}_{A_\gamma} := \Id_{A_1^n} - \hat{P}^{ x_\gamma, \theta_\gamma}_{A_\gamma}
\end{align}
to be the projectors on the subspace with relative weight at most $\epsilon + \epsilon_m$, and at least $\epsilon + \epsilon_m$ with respect to $x_\gamma$ in the basis $\theta_\gamma$. Here the parameter $\epsilon$ is the same as the source error threshold in the previous section and $\epsilon_m>0$ is a small parameter quantifying the measurement device error. The projector $\hat{P}^{ x_\gamma, \theta_\gamma}_{A_\gamma}$ is the rotated version of projector $\hat{P}$, which was used for the measurement map in the previous section. In this section, we need to use the rotated version because the real measurements in an implementation will depend on the inputs $\gamma, x_\gamma$ and $\theta_\gamma$. \\

\begin{sloppypar}
We assume that for some fixed small $\xi >0$ the measurement elements $\{\Lambda\rndBrk{\leq \epsilon | \gamma, x_\gamma , \theta_\gamma}_{A_\gamma}\}_{\gamma, x_\gamma, \theta_\gamma}$ satisfy the following for every collection of states $\{\sigma^{(\gamma)}_{A_\gamma| x_\gamma \theta_\gamma}\}_{\gamma, x_\gamma, \theta_\gamma}$:
\end{sloppypar}
\begin{align}
    \sum_{\gamma}p(\gamma) \sum_{x_\gamma, \theta_\gamma} p(x_\gamma, \theta_\gamma) \tr\rndBrk{\Lambda\rndBrk{\leq \epsilon | \gamma, x_\gamma , \theta_\gamma}_{A_\gamma} \hat{P}_{A_\gamma}^{\perp| x_\gamma, \theta_\gamma} \sigma^{(\gamma)}_{A_\gamma| x_\gamma \theta_\gamma}  \hat{P}_{A_\gamma}^{\perp| x_\gamma, \theta_\gamma}} \leq \xi.
    \label{eq:assmpn_imp_meas}
\end{align}
Stated in words, we require that for any collection of states $\{\sigma^{(\gamma)}_{A_\gamma| x_\gamma \theta_\gamma}\}_{\gamma, x_\gamma, \theta_\gamma}$ with a relative weight larger than $\epsilon + \epsilon_m$ (lying in the subspace corresponding to the projector $\hat{P}_{A_\gamma}^{\perp| x_\gamma, \theta_\gamma}$), the probability that a weight lesser than $\epsilon$ is measured is smaller than $\xi$ when averaged over the choice of the random set $\gamma$ and $x_\gamma, \theta_\gamma$. Using this assumption on the measurements, in Lemma \ref{lemm:smoothDmax_bd_imp_meas} we will derive a smooth max-relative entropy bound similar to the one in the previous section. The smoothing parameter of the relative entropy in this bound, however, will depend on $\xi$, which in turn implies that the privacy amplification error of the subsequent QKD protocol will be lower bounded by a function of $\xi$. It does not seem that this dependence of the smoothing parameter on $\xi$ can be avoided. For example, if the measurements measure a small weight for a set of large weight states and the source emits those states, then they can be exploited by Eve to extract additional information during the QKD protocol. It also seems that we cannot use some kind of joint test for the source and measurement device (similar to Protocol \ref{frame:src_corr_prot}) without an additional assumption to ensure that the weight measured by the measurement device is almost correct, since the source can always embed its information using an arbitrary unitary and the measurement can always decode that information using the same unitary. \\

I.I.D measurements with error $\epsilon'_m$ or more generally measurements, which are guaranteed to measure each input qubit correctly with probability at least $(1- \epsilon'_m)$ independent of the previous rounds (both these examples consider measurements which measure the qubits $A_\gamma$ in the provided basis $\theta_\gamma$ to produce the results $\hat{x}_\gamma$ and then use these results to test if $\omega(x_\gamma \oplus \hat{x}_\gamma) \leq \epsilon$ or not), satisfy the above assumption for the choice of some $\delta' >0$, $\epsilon_m = \epsilon'_m + \delta'$ and $\xi = e^{-2m{\delta'}^2}$ (using the Chernoff-Hoeffding bound). Additionally, since we average over the random set $\gamma$ as well, it is possible to guarantee with high probability that for most test measurements the relaxation time of the measurement device is large. This should enable us to model a large and practical class of measurements using these assumptions. We leave the details for the specific measurement model for future work.\\ 

We will show that for measurements, which satisfy the above assumption the following Lemma holds. One can use this bound in place Lemma \ref{lemm:smoothDmax_bd_perf_meas} to prove a smooth min-entropy lower bound for the QKD protocol, similar to the previous section. Note that the following proof builds on the proof of Lemma \ref{lemm:smoothDmax_bd_perf_meas} and makes use of definitions used in that proof. 

\begin{lemma}
    \label{lemm:smoothDmax_bd_imp_meas}
    Suppose that the measurements used for the source test $\{\Lambda\rndBrk{\leq \epsilon | \gamma, x_\gamma , \theta_\gamma}_{A_\gamma}\}_{\gamma, x_\gamma, \theta_\gamma}$ satisfy the assumption in Eq. \ref{eq:assmpn_imp_meas} with parameters $\epsilon_{m}$ and $\xi$. Let $\epsilon$ be the threshold of the source test, $\delta \in (0,1)$ a small parameter, and let $\hat{\rho}_{X \Theta A}^{(\epsilon+\epsilon_m +\delta)} := (1 - 2(\epsilon +\epsilon_m +\delta)) \hat{\rho}_{X \Theta A} + 2(\epsilon+\epsilon_m+ \delta) \hat{\rho}_{X \Theta}\otimes \tau_A$ where $\tau_A$ is the completely mixed state on the register $A$. Let the event $\Omega_{\text{im}}$ denote that the source test using the imperfect measurements succeeds and let state $\rho'_{X_1^{n} \Theta_1^{n} A_1^{n}|\Omega_{\text{im}}}$ denote the state produced by the source test conditioned on passing. For this state, we have that
    \begin{align*}
        D_{\max}^{\epsilon_{f}} (\rho'_{X_1^n \Theta_1^n A_1^n | \Omega_{\text{im}}}|| \rndBrk{\hat{\rho}_{X \Theta A}^{(\epsilon+\epsilon_m +\delta)}}^{\otimes n}) &\leq n h(\epsilon+\epsilon_m +\delta) + 2 + \log\frac{1}{\Pr_{\rho}( \Omega_{\text{im}}) - \epsilon^{\delta}_{\text{qu}}} \\
        &+ \log\frac{1}{4\xi (\Pr_{\rho}( \Omega_{\text{im}}) - \epsilon^{\delta}_{\text{qu}} - 4\xi)}
        \numberthis
    \end{align*}
    where $\Pr_{\rho}(\Omega_{\text{im}})$ is the probability of the event $\Omega_{\text{im}}$ when the testing procedure is applied to the state $\rho$, $h(x) = - x \log(x) - (1-x)\log (1-x)$ is the binary entropy function and $\epsilon_{f} := \frac{2\xi^{1/2}}{\sqrt{\Pr_{{\rho}}( \Omega_{\text{im}})- \epsilon^{\delta}_{\text{qu}}}} + 2\sqrt{\frac{\epsilon^{\delta}_{\text{qu}}}{\Pr_{\rho}( \Omega_{\text{im}})}}$ for $\epsilon_{\text{qu}}^{\delta} = \sqrt{2}\exp\rndBrk{- \frac{n \delta^2}{2(n+2)} m }$
\end{lemma}

\begin{proof}
    For every $x_1^{n+m}$ and $\theta_1^{n+m}$, we define the following appropriately rotated versions of the projector $\Pi^{\delta | \gamma}_{A_1^{n+m}}$ given by Theorem \ref{th:qu_sampling}, so that we can compare the relative weight with the string $x_1^{n+m}$ in the basis given by $\theta_1^{n+m}$.
    \begin{align}
        \bar{\Pi}^{\delta | \gamma, x_1^{n+m}, \theta_1^{n+m}}_{A_1^{n+m}} := \bigotimes_{i=1}^{n+m} V^{x_i, \theta_i}_{A_i} \Pi^{\delta | \gamma}_{A_1^{n+m}} \bigotimes_{i=1}^{n+m} (V^{x_i, \theta_i}_{A_i})^{\dagger}
    \end{align}
    where the unitaries $V^{x_i, \theta_i}_{A_i}$ are defined in Eq. \ref{eq:V_unit_defn}. We use the state $\eta$ from the previous section (Eq. \ref{eq:eta_defn}) to define the state 
    \begin{align}
        \tilde{\rho}_{\Gamma X_1^{n+m} \Theta_1^{n+m} A_1^{n+m}} := \bigotimes_{i=1}^{n+m} V_{X_i \Theta_i A_i} \eta_{\Gamma X_1^{n+m} \Theta_1^{n+m} A_1^{n+m}} \bigotimes_{i=1}^{n+m} V^\dagger_{X_i \Theta_i A_i}.
        \label{eq:defn_tilde_rho}
    \end{align}
    Using the distance bound proven in Eq. \ref{eq:nu_eta_dist_bd} and the definition of $\nu$ in Eq. \ref{eq:defn_nu}, we have
    \begin{align*}
        \frac{1}{2}\norm{\rho_{\Gamma X_1^{n+m} \Theta_1^{n+m} A_1^{n+m}} - \tilde{\rho}_{\Gamma X_1^{n+m} \Theta_1^{n+m} A_1^{n+m}}}_1 \leq \epsilon^{\delta}_{\text{qu}}
        \numberthis
    \end{align*}
    The conditional states $\tilde{\rho}^{(\gamma)}_{A_1^{n+m}| x_1^{n+m} \theta_1^{n+m}}$ of the state $\tilde{\rho}$ above satisfy
    \begin{align*}
        \bar{\Pi}&^{\delta | \gamma, x_1^{n+m}, \theta_1^{n+m}}_{A_1^{n+m}} \tilde{\rho}^{(\gamma)}_{A_1^{n+m}| x_1^{n+m} \theta_1^{n+m}} \bar{\Pi}^{\delta | \gamma, x_1^{n+m}, \theta_1^{n+m}}_{A_1^{n+m}} \\
        &= \bigotimes_{i=1}^{n+m} V^{x_i, \theta_i}_{A_i} \Pi^{\delta | \gamma}_{A_1^{n+m}} \bigotimes_{i=1}^{n+m} (V^{x_i, \theta_i}_{A_i})^{\dagger} \tilde{\rho}^{(\gamma)}_{A_1^{n+m}| x_1^{n+m} \theta_1^{n+m}} \bigotimes_{i=1}^{n+m} V^{x_i, \theta_i}_{A_i} \Pi^{\delta | \gamma}_{A_1^{n+m}} \bigotimes_{i=1}^{n+m} (V^{x_i, \theta_i}_{A_i})^{\dagger} \\
        &= \bigotimes_{i=1}^{n+m} V^{x_i, \theta_i}_{A_i} \Pi^{\delta | \gamma}_{A_1^{n+m}} {\eta}^{(\gamma)}_{A_1^{n+m}| x_1^{n+m} \theta_1^{n+m}} \Pi^{\delta | \gamma}_{A_1^{n+m}} \bigotimes_{i=1}^{n+m} (V^{x_i, \theta_i}_{A_i})^{\dagger} \\
        &= \bigotimes_{i=1}^{n+m} V^{x_i, \theta_i}_{A_i} {\eta}^{(\gamma)}_{A_1^{n+m}| x_1^{n+m} \theta_1^{n+m}} \bigotimes_{j=1}^{n+m} (V^{x_i, \theta_i}_{A_i})^{\dagger} \\
        &= \tilde{\rho}^{(\gamma)}_{A_1^{n+m}| x_1^{n+m} \theta_1^{n+m}}
        \numberthis
    \end{align*} 
    where we have used the definition of $\tilde{\rho}^{(\gamma)}_{A_1^{n+m}| x_1^{n+m} \theta_1^{n+m}}$ (Eq. \ref{eq:defn_tilde_rho}) in the second equality, and Eq. \ref{eq:cond_eta_lies_in_projector} for the fourth line. \\

    We call the subnormalised state produced after performing the (imperfect) measurements on the states $\rho_{\Gamma X_1^{n+m} \Theta_1^{n+m} A_1^{n+m}}$, conditioning on the event $\Omega_{\text{im}}$ and tracing over the registers $X_{\Gamma}$ and $\Theta_{\Gamma}$ as ${\rho}'_{{\Gamma} X_{\bar{{\Gamma}}} \Theta_{\bar{{\Gamma}}} A_{\bar{{\Gamma}}} \wedge \Omega_{\text{im}}}$. Similarly, we let $\tilde{\rho}'_{{\Gamma} X_{\bar{{\Gamma}}} \Theta_{\bar{{\Gamma}}} A_{\bar{{\Gamma}}} \wedge \Omega_{\text{im}}}$ denote the subnormalised state produced when this subnormalised map is applied to $\tilde{\rho}_{\Gamma X_1^{n+m} \Theta_1^{n+m} A_1^{n+m}}$. We have that 
    \begin{align*}
        \tilde{\rho}'_{{\Gamma} X_{\bar{{\Gamma}}} \Theta_{\bar{{\Gamma}}} A_{\bar{{\Gamma}}} \wedge \Omega_{\text{im}}} &= \sum_{\gamma, x_{\bar{\gamma}}, \theta_{\bar{\gamma}}} p(\gamma) p(x_{\bar{\gamma}}, \theta_{\bar{\gamma}}) \ket{\gamma, x_{\bar{\gamma}}, \theta_{\bar{\gamma}}} \bra{\gamma, x_{\bar{\gamma}}, \theta_{\bar{\gamma}}} \otimes \\
        &\qquad \qquad \sum_{x_\gamma, \theta_\gamma} p(x_\gamma, \theta_\gamma) \tr_{A_\gamma} \rndBrk{\Lambda\rndBrk{\leq \epsilon | \gamma, x_\gamma , \theta_\gamma}_{A_\gamma} \tilde{\rho}^{(\gamma)}_{A_\gamma A_{\bar{\gamma}}| x_1^{n+m} \theta_1^{n+m}}} \\
        &\leq 2 \sum_{\gamma, x_{\bar{\gamma}}, \theta_{\bar{\gamma}}} p(\gamma) p(x_{\bar{\gamma}}, \theta_{\bar{\gamma}}) \ket{\gamma, x_{\bar{\gamma}}, \theta_{\bar{\gamma}}} \bra{\gamma, x_{\bar{\gamma}}, \theta_{\bar{\gamma}}} \otimes \\
        &\qquad \qquad \bigg[\sum_{x_\gamma, \theta_\gamma} p(x_\gamma, \theta_\gamma) \tr_{A_\gamma} \rndBrk{\Lambda\rndBrk{\leq \epsilon | \gamma, x_\gamma , \theta_\gamma}_{A_\gamma} \hat{P}^{ x_\gamma, \theta_\gamma}_{A_\gamma} \tilde{\rho}^{(\gamma)}_{A_\gamma A_{\bar{\gamma}}| x_1^{n+m} \theta_1^{n+m}} \hat{P}^{ x_\gamma, \theta_\gamma}_{A_\gamma}} \\
        &\qquad \qquad + \sum_{x_\gamma, \theta_\gamma} p(x_\gamma, \theta_\gamma) \tr_{A_\gamma} \rndBrk{\Lambda\rndBrk{\leq \epsilon | \gamma, x_\gamma , \theta_\gamma}_{A_\gamma} \hat{P}^{\perp|  x_\gamma, \theta_\gamma}_{A_\gamma} \tilde{\rho}^{(\gamma)}_{A_\gamma A_{\bar{\gamma}}| x_1^{n+m} \theta_1^{n+m}} \hat{P}^{\perp|  x_\gamma, \theta_\gamma}_{A_\gamma}}\bigg] \\
        &\leq 2 \sum_{\gamma, x_{\bar{\gamma}}, \theta_{\bar{\gamma}}} p(\gamma) p(x_{\bar{\gamma}}, \theta_{\bar{\gamma}}) \ket{\gamma, x_{\bar{\gamma}}, \theta_{\bar{\gamma}}} \bra{\gamma, x_{\bar{\gamma}}, \theta_{\bar{\gamma}}} \otimes \\
        &\qquad \qquad \sum_{x_\gamma, \theta_\gamma} p(x_\gamma, \theta_\gamma) \tr_{A_\gamma} \rndBrk{\Lambda\rndBrk{\leq \epsilon | \gamma, x_\gamma , \theta_\gamma}_{A_\gamma} \hat{P}^{ x_\gamma, \theta_\gamma}_{A_\gamma} \tilde{\rho}^{(\gamma)}_{A_\gamma A_{\bar{\gamma}}| x_1^{n+m} \theta_1^{n+m}} \hat{P}^{ x_\gamma, \theta_\gamma}_{A_\gamma}} \\
        &\qquad \qquad + 2\xi \mu_{\Gamma X_{\bar{\Gamma}} \Theta_{\bar{\Gamma}} A_{\bar{\Gamma}}}
        \numberthis
    \end{align*}
    where we have used the pinching inequality (see, for example \cite[Section 2.6.3]{TomamichelBook16}) in the second line, defined the state $\mu_{\Gamma X_{\bar{\Gamma}} \Theta_{\bar{\Gamma}} A_{\bar{\Gamma}}}$ as the normalization of the state
    \begin{align*}
        &\sum_{\gamma, x_{\bar{\gamma}}, \theta_{\bar{\gamma}}} p(\gamma) p(x_{\bar{\gamma}}, \theta_{\bar{\gamma}}) \ket{\gamma, x_{\bar{\gamma}}, \theta_{\bar{\gamma}}} \bra{\gamma, x_{\bar{\gamma}}, \theta_{\bar{\gamma}}} \otimes \\
        &\qquad \qquad \sum_{x_\gamma, \theta_\gamma} p(x_\gamma, \theta_\gamma) \tr_{A_\gamma} \rndBrk{\Lambda\rndBrk{\leq \epsilon | \gamma, x_\gamma , \theta_\gamma}_{A_\gamma} \hat{P}^{\perp|  x_\gamma, \theta_\gamma}_{A_\gamma} \tilde{\rho}^{(\gamma)}_{A_\gamma A_{\bar{\gamma}}| x_1^{n+m} \theta_1^{n+m}} \hat{P}^{\perp|  x_\gamma, \theta_\gamma}_{A_\gamma}}
        \numberthis
    \end{align*}
    and used 
    \begin{align*}
        &\tr\bigg( \sum_{\gamma, x_{\bar{\gamma}}, \theta_{\bar{\gamma}}} p(\gamma) p(x_{\bar{\gamma}}, \theta_{\bar{\gamma}}) \ket{\gamma, x_{\bar{\gamma}}, \theta_{\bar{\gamma}}} \bra{\gamma, x_{\bar{\gamma}}, \theta_{\bar{\gamma}}} \otimes \\
        &\qquad \qquad \sum_{x_\gamma, \theta_\gamma} p(x_\gamma, \theta_\gamma) \tr_{A_\gamma} \rndBrk{\Lambda\rndBrk{\leq \epsilon | \gamma, x_\gamma , \theta_\gamma}_{A_\gamma} \hat{P}^{\perp|  x_\gamma, \theta_\gamma}_{A_\gamma} \tilde{\rho}^{(\gamma)}_{A_\gamma A_{\bar{\gamma}}| x_1^{n+m} \theta_1^{n+m}} \hat{P}^{\perp|  x_\gamma, \theta_\gamma}_{A_\gamma}}\bigg) \\
        &= \sum_{\gamma, x_{\bar{\gamma}}, \theta_{\bar{\gamma}}} p(\gamma) p(x_{\bar{\gamma}}, \theta_{\bar{\gamma}}) \sum_{x_\gamma, \theta_\gamma} p(x_\gamma, \theta_\gamma) \tr \rndBrk{\Lambda\rndBrk{\leq \epsilon | \gamma, x_\gamma , \theta_\gamma}_{A_\gamma} \hat{P}^{\perp|  x_\gamma, \theta_\gamma}_{A_\gamma} \tilde{\rho}^{(\gamma)}_{A_\gamma A_{\bar{\gamma}} | x_1^{n+m} \theta_1^{n+m}} \hat{P}^{\perp|  x_\gamma, \theta_\gamma}_{A_\gamma}} \\
        &= \sum_{\gamma} p(\gamma) \sum_{x_\gamma, \theta_\gamma} p(x_\gamma, \theta_\gamma) \tr \rndBrk{\Lambda\rndBrk{\leq \epsilon | \gamma, x_\gamma , \theta_\gamma}_{A_\gamma} \hat{P}^{\perp|  x_\gamma, \theta_\gamma}_{A_\gamma} \rndBrk{\sum_{x_{\bar{\gamma}}, \theta_{\bar{\gamma}}} p(x_{\bar{\gamma}}, \theta_{\bar{\gamma}}) \tilde{\rho}^{(\gamma)}_{A_\gamma | x_1^{n+m} \theta_1^{n+m}}} \hat{P}^{\perp|  x_\gamma, \theta_\gamma}_{A_\gamma}} \\
        &\leq \xi,
        \numberthis
    \end{align*}
    which follows from our assumption about the measurements (Eq. \ref{eq:assmpn_imp_meas}). Therefore, we have 
    \begin{align*}
        \tilde{\rho}'_{{\Gamma} X_{\bar{{\Gamma}}} \Theta_{\bar{{\Gamma}}} A_{\bar{{\Gamma}}}  \wedge \Omega_{\text{im}}} &\leq 2 \sum_{\gamma, x_{\bar{\gamma}}, \theta_{\bar{\gamma}}} p(\gamma) p(x_{\bar{\gamma}}, \theta_{\bar{\gamma}}) \ket{\gamma, x_{\bar{\gamma}}, \theta_{\bar{\gamma}}} \bra{\gamma, x_{\bar{\gamma}}, \theta_{\bar{\gamma}}} \otimes \\
        &\qquad \qquad \sum_{x_\gamma, \theta_\gamma} p(x_\gamma, \theta_\gamma) \tr_{A_\gamma} \rndBrk{\Lambda\rndBrk{\leq \epsilon | \gamma, x_\gamma , \theta_\gamma}_{A_\gamma} \hat{P}^{ x_\gamma, \theta_\gamma}_{A_\gamma} \tilde{\rho}^{(\gamma)}_{A_\gamma A_{\bar{\gamma}}| x_1^{n+m} \theta_1^{n+m}} \hat{P}^{ x_\gamma, \theta_\gamma}_{A_\gamma}} \\
        &\qquad \qquad+ 2\xi \mu_{\Gamma X_{\bar{\Gamma}} \Theta_{\bar{\Gamma}} A_{\bar{\Gamma}}}\\
        &\leq 2 \sum_{\gamma, x_{\bar{\gamma}}, \theta_{\bar{\gamma}}} p(\gamma) p(x_{\bar{\gamma}}, \theta_{\bar{\gamma}}) \ket{\gamma, x_{\bar{\gamma}}, \theta_{\bar{\gamma}}} \bra{\gamma, x_{\bar{\gamma}}, \theta_{\bar{\gamma}}} \otimes \\
        &\qquad \qquad \sum_{x_\gamma, \theta_\gamma} p(x_\gamma, \theta_\gamma) \tr_{A_\gamma} \rndBrk{\hat{P}^{ x_\gamma, \theta_\gamma}_{A_\gamma} \tilde{\rho}^{(\gamma)}_{A_\gamma A_{\bar{\gamma}}| x_1^{n+m} \theta_1^{n+m}}} + 2\xi \mu_{\Gamma X_{\bar{\Gamma}} \Theta_{\bar{\Gamma}} A_{\bar{\Gamma}}}\\
        & = 2 \bar{\tilde{\rho}}_{\Gamma X_{\bar{\Gamma}} \Theta_{\bar{\Gamma}} A_{\bar{\Gamma}} \wedge \Omega}^{(\epsilon + \epsilon_m)} + 2\xi \mu_{\Gamma X_{\bar{\Gamma}} \Theta_{\bar{\Gamma}} A_{\bar{\Gamma}}}
        \numberthis
    \end{align*}
    where the state $\bar{\tilde{\rho}}_{\Gamma X_{\bar{\Gamma}} \Theta_{\bar{\Gamma}} A_{\bar{\Gamma}} \wedge \Omega}^{(\epsilon + \epsilon_m)}$ is the state produced when the perfect measurement is used to measure $A_\gamma$ and condition the state $\tilde{\rho}_{{\Gamma} X_1^{n+m} \Theta_1^{n+m} A_1^{n+m}}$ on the event that the relative weight of the measured results is lesser than $\epsilon+\epsilon_m$ from the string contained in $X_{\gamma}$. This is the state, which was used in the previous section to derive the smooth max-relative entropy bound. The only difference being that the threshold for the relative weight of the perfect measurement in the last section was $\epsilon$. Thus, we can use the previously derived bound in Eq. \ref{eq:conditioned_rho_tilde_Dmax} for this state by simply replacing $\epsilon$ with $\epsilon+\epsilon_m$. Relabelling the remaining registers between $1$ and $n$, tracing over the $\Gamma$ register and using the Eq. \ref{eq:conditioned_rho_tilde_Dmax}, we get 
    \begin{align*}
        \tilde{\rho}'_{X_1^n \Theta_1^n A_1^n  \wedge \Omega_{\text{im}}} &\leq 2 \bar{\tilde{\rho}}_{X_1^n \Theta_1^n A_1^n  \wedge \Omega}^{(\epsilon + \epsilon_m)} + 2\xi \mu_{ X_1^n \Theta_1^n A_1^n } \\
        &\leq 2^{n h(\epsilon+\epsilon_m +\delta) + 1}\rndBrk{\hat{\rho}_{X \Theta A}^{(\epsilon+\epsilon_m +\delta)}}^{\otimes n} + 2\xi \mu_{X_1^n \Theta_1^n A_1^n } \numberthis
    \end{align*}
    where $\hat{\rho}_{X \Theta A}^{(\epsilon+\epsilon_m +\delta)} := (1-2(\epsilon+\epsilon_m +\delta))\hat{\rho}_{X \Theta A} + 2(\epsilon+\epsilon_m +\delta)\hat{\rho}_{X \Theta}\otimes \tau_A$. As before using the data processing inequality, we have 
    \begin{align}
        \frac{1}{2} \norm{{\rho}'_{X_1^n \Theta_1^n A_1^n \wedge \Omega_{\text{im}}} - \tilde{\rho}'_{X_1^n \Theta_1^n A_1^n  \wedge \Omega_{\text{im}}}}_1 \leq \epsilon^\delta_{\text{qu}}.
    \end{align}
    Once again following the argument in \cite[Lemma G.1]{Marwah23}, the conditional states satisfy
    \begin{align}
        \frac{1}{2} \norm{{\rho}'_{X_1^n \Theta_1^n A_1^n | \Omega_{\text{im}}} -\tilde{\rho}'_{X_1^n \Theta_1^n A_1^n | \Omega_{\text{im}}}}_1 \leq \frac{2\epsilon^\delta_{\text{qu}}}{\Pr_{\rho}(\Omega_{\text{im}})}
    \end{align}
    for $\Pr_{\rho}(\Omega_{\text{im}}) := \tr\rndBrk{{\rho}'_{X_1^n \Theta_1^n A_1^n \wedge \Omega_{\text{im}}}}$, defined as the probability that the Protocol \ref{frame:src_corr_prot} does not abort with the imperfect measurements and 
    \begin{align}
        \tilde{\rho}'_{X_1^n \Theta_1^n A_1^n | \Omega_{\text{im}}} &\leq \frac{2^{n h(\epsilon+\epsilon_m +\delta) + 1}}{\Pr_{\tilde{\rho}}( \Omega_{\text{im}})}\rndBrk{\hat{\rho}_{X \Theta A}^{(\epsilon+\epsilon_m +\delta)}}^{\otimes n} + \frac{4 \xi}{{\Pr_{\tilde{\rho}}( \Omega_{\text{im}})}} \frac{\mu_{X_1^n \Theta_1^n A_1^n}}{2}.
        \label{eq:meas_tilde_operator_ineq}
    \end{align}
    where $\Pr_{\tilde{\rho}}(\Omega_{\text{im}}) := \tr\rndBrk{{\tilde{\rho}}'_{X_1^n \Theta_1^n A_1^n \wedge \Omega_{\text{im}}}}$. For $0< \mu <1 $, the hypothesis testing relative entropy \cite{Wang12} is defined as 
    \begin{align}
        D_h^{\mu}(\rho || \sigma) := -\inf \curlyBrk{\log \tr(\sigma Q): 0\leq \mu Q \leq \Id,\text{ and }  \tr(\rho Q)\geq 1}.
    \end{align}
    Equivalently, using semidefinite programming duality (see \cite{JW-Adv-QIT_HypTest}) it can be shown that 
    \begin{align}
        D_h^{\mu}(\rho || \sigma) &= -\sup \curlyBrk{\log(\lambda - \tr(Y)): Y\geq 0, \lambda \geq 0, \text{ and } \lambda \rho \leq \sigma + \mu Y}\\
        &= \inf\curlyBrk{\log \lambda' - \log (1- \tr(Z)): Z\geq 0, \lambda' \geq 0, \text{ and } \rho \leq \lambda' \sigma + \mu Z}.
    \end{align}
    Thus, Eq. \ref{eq:meas_tilde_operator_ineq} implies 
    \begin{align}
        D_h^{\mu}(\tilde{\rho}'_{X_1^n \Theta_1^n A_1^n | \Omega_{\text{im}}}|| (\hat{\rho}_{X \Theta A}^{(\epsilon+\epsilon_m +\delta)})^{\otimes n}) \leq n h(\epsilon+\epsilon_m +\delta) + 2 + \log\frac{1}{\Pr_{\tilde{\rho}}( \Omega_{\text{im}})}
    \end{align}
    for $\mu := \frac{4\xi}{{\Pr_{\tilde{\rho}}( \Omega_{\text{im}})}}$. Using \cite[Theorem 5.11]{JW-Adv-QIT_HypTest} (originally proven in \cite{Anshu19}), this implies that\footnote{The smoothing for $D^{\epsilon}_{\max}(\rho || \sigma)$ in \cite{JW-Adv-QIT_HypTest} is defined using the trace distance instead of purified distance, which we use here. It can, however, be verified that the proof there also works with purified distance.}
    \begin{align}
        D_{\max}^{\sqrt{\mu}} (\tilde{\rho}'_{X_1^n \Theta_1^n A_1^n | \Omega_{\text{im}}}|| (\hat{\rho}_{X \Theta A}^{(\epsilon+\epsilon_m +\delta)})^{\otimes n}) \leq n h(\epsilon+\epsilon_m +\delta) + 2 + \log\frac{1}{\Pr_{\tilde{\rho}}( \Omega_{\text{im}})} + \log\frac{1}{\mu (1- \mu)}
    \end{align}
    Using the triangle inequality, we can state this in terms of the real state $\tilde{\rho}'_{X_1^n \Theta_1^n A_1^n | \Omega_{\text{im}}}$
    \begin{align*}
        D_{\max}^{\epsilon_{f}} (\rho'_{X_1^n \Theta_1^n A_1^n | \Omega_{\text{im}}}|| \rndBrk{\hat{\rho}_{X \Theta A}^{(\epsilon+\epsilon_m +\delta)}}^{\otimes n}) &\leq n h(\epsilon+\epsilon_m +\delta) + 2 + \log\frac{1}{\Pr_{\rho}( \Omega_{\text{im}}) - \epsilon^{\delta}_{\text{qu}}} \\
        &+ \log\frac{1}{4\xi (\Pr_{\rho}( \Omega_{\text{im}}) - \epsilon^{\delta}_{\text{qu}} - 4\xi)}
        \numberthis
    \end{align*}
    for $\epsilon_{f} := \frac{2\xi^{1/2}}{\sqrt{\Pr_{{\rho}}( \Omega_{\text{im}})- \epsilon^{\delta}_{\text{qu}}}} + 2\sqrt{\frac{\epsilon^{\delta}_{\text{qu}}}{\Pr_{\rho}( \Omega_{\text{im}})}}$. Note that if $\xi = \exp(-\Omega(m))$, then the last term in the bound above adds $O(m)$ to the smooth max-relative entropy, so it cannot be chosen to be too small (This seems to be an artifact of the bound in \cite[Theorem 5.11]{JW-Adv-QIT_HypTest}, and it seems that it should be possible to improve this dependence).
\end{proof}
 
\section{Discussion and future work}
\label{sec:src_cor_disc}

We demonstrate a general method to reduce the security of the BB84 protocol with an imperfect source with source correlations to that of the BB84 protocol with an almost perfect source. In order to minimise the rate loss and privacy amplification error, we use a source test to test the output of the imperfect source before using it for the QKD protocol. Theorem \ref{th:src_cor_Hmin_bd} gives a simple bound on the smooth min-entropy for the BB84 protocol which uses the output of the source test. According to this bound, for a source error of $\epsilon$, the rate of the QKD protocol decreases by $ O((\epsilon \log\frac{1}{\epsilon})^{1/2})$ and the privacy amplification error can be made arbitrarily small assuming perfect measurements are used for the source test. With imperfect measurements, satisfying a very broad assumption, we showed that the rate decrease is similar to the perfect case and the privacy amplification error depends on an error parameter of the measurements. This error parameter too can be made arbitrarily small under further reasonable physical assumptions, like independence of the measurement errors or almost perfect behaviour given a sufficient relaxation time. We leave the details of such a physical model and its relation to our assumption on the measurements for future work. It should be noted that one could also place physical assumptions on the source, which would guarantee that it passes the source test and hence imply security for the protocol. Further, if the source can be guaranteed to pass the source test with a high probability (which can be made arbitrarily close to $1$), say $1- \epsilon_s$, then the source test need not even be performed before the QKD protocol. The error $\epsilon_s$ can simply be added to the QKD security parameter.\\

While we have provided an effective and general method for proving security under source correlations for a theoretical BB84 protocol, further work is required for handling source correlations for practical QKD protocols. This includes extending our analysis to protocols used in practice, and incorporating other known device imperfections \cite{Pereira22}. As mentioned at the end of Section \ref{sec:sec_proof_src_corr}, further work is also needed to understand the behaviour of the rate loss on the protocol parameters.

\section*{Acknowledgments}
We would like to thank Ernest Tan and Shlok Ashok Nahar, who explained the source correlation problem for QKD to us during and after the QKD Security Proof Workshop at the Institute of Quantum Computing, Waterloo and also for their comments on a manuscript of this paper. AM was supported by the J.A. DeS\`eve Foundation and by bourse d'excellence Google. This work was also supported by the Natural Sciences and Engineering Research Council of Canada.

\appendix

\addcontentsline{toc}{section}{APPENDICES}
\section*{APPENDICES}

\section{Proof of Theorem \ref{th:src_cor_Hmin_bd}}
\label{sec:formal_src_qkd_pf}

In this section, we formally prove the lower bound on the smooth min-entropy required for the security of QKD in Theorem \ref{th:src_cor_Hmin_bd} using the entropy accumulation theorem (EAT). In Section \ref{sec:sec_proof_src_corr} (Eq. \ref{eq:src_corr_dist_bd} and \ref{eq:src_corr_Dmax_bd}), we showed that $\rho'_{X_1^n \Theta_1^n A_1^n} := \bar{\tilde{\rho}}_{X_1^n \Theta_1^n A_1^n|\Omega}$ and $\sigma_{X_1^n \Theta_1^n A_1^n} = \rndBrk{\hat{\rho}_{X \Theta A}^{(\epsilon +\delta)}}^{\otimes n}$ is such that
\begin{align}
    \frac{1}{2}\norm{\rho'_{X_1^n \Theta_1^n A_1^n} - \bar{\rho}_{X_1^n \Theta_1^n A_1^n | \Omega}}_1 \leq \frac{\epsilon_f^2}{2}
    \label{eq:rho_prime_dist_bd}
\end{align}
and
\begin{align}
    D_{\max} (\rho'_{X_1^{n} \Theta_1^{n} A_1^{n} } || \sigma_{X_1^n \Theta_1^n A_1^n})  \leq n h(\epsilon + \delta) + \log \frac{1}{\Pr_{{\rho}} (\Omega) - \epsilon_{\text{qu}}^{\delta}}.
    \label{eq:rho_prime_Dmax_bound}
\end{align}
Fix an arbitrary strategy for Eve. Let $\Phi_{\text{QKD}} : X_1^n \Theta_1^n A_1^n \rightarrow X_1^n Y_1^n \hat{X}_S \hat{C}_1^n \Theta_1^n \hat{\Theta}_1^n S T E$ be the map applied by Alice, Bob and Eve on the states produced by Alice during the QKD protocol. In order to prove security for the BB84 protocol, we need a lower bound on the following smooth min-entropy of $\Phi_{\text{QKD}} (\bar{\rho})$
\begin{align*}
    H^{\nu}_{\min} (X_S | E T \Theta_1^n \hat{\Theta}_1^n)_{\Phi_{\text{QKD}} (\bar{\rho})_{|\Upsilon}}
\end{align*}
for some $\nu \geq 0$. In \cite[Appendix A]{Metger22-2}, it is shown that it is sufficient to show a lower bound for the smooth min-entropy of the final state of the protocol conditioned on the event $\Upsilon''$ when the protocol uses perfect source states. The arguments mentioned there are also valid for our case, which is why we bound the smooth min-entropy
\begin{align*}
    H^{\nu}_{\min} (X_S | E T \Theta_1^n \hat{\Theta}_1^n)_{\Phi_{\text{QKD}} (\bar{\rho})_{|\Upsilon''}}
\end{align*}
in Theorem \ref{th:src_cor_Hmin_bd}\footnote{\label{fn:obs_sec_cond}The arguments in \cite[Appendix A]{Metger22-2} can also be modified to show that it is sufficient to show that $P(\Upsilon'') \norm{\rho^f_{K_A E'} - \tau_{K_A} \otimes \rho^f_{E'}}_1$ is small, where $K_A$ is Alice's key and $\rho^f$ is the state produced at the end of the protocol conditioned on not aborting, to prove the security of QKD.}. \\

Using the data processing inequality and Eq. \ref{eq:rho_prime_Dmax_bound}, we see that
\begin{align}
    D_{\max} (\Phi_{\text{QKD}}(\rho'_{X_1^{n} \Theta_1^{n} A_1^{n} }) || \Phi_{\text{QKD}}(\sigma_{X_1^n \Theta_1^n A_1^n}))  \leq n h(\epsilon + \delta) + \log \frac{1}{\Pr_{{\rho}} (\Omega) - \epsilon_{\text{qu}}^{\delta}}.
    \label{eq:src_corr_tilde_rho_sigma_Dmax_bd}
\end{align}
Note that $\Phi_{\text{QKD}} (\rho'_{X_1^n \Theta_1^n A_1^n})$ and $\Phi_{\text{QKD}}(\sigma_{X_1^n \Theta_1^n A_1^n})$ are the states that are produced at the end of the protocol if Alice's source were to produce the states $\rho'_{X_1^n \Theta_1^n A_1^n}$ and $\sigma_{X_1^n \Theta_1^n A_1^n}$ respectively. The states $\Phi_{\text{QKD}} (\rho'_{X_1^n \Theta_1^n A_1^n})$ and $\Phi_{\text{QKD}}(\sigma_{X_1^n \Theta_1^n A_1^n})$ also contain all the corresponding classical variables as the real protocol state $\Phi_{\text{QKD}} (\bar{\rho}_{X_1^n \Theta_1^n A_1^n | \Omega})$. In particular, the event $\Upsilon''$ is well-defined (defined using classical variables) for both of these states. \\

Using \cite[Lemma G.1]{Marwah23} and Eq. \ref{eq:rho_prime_dist_bd}, we have that the final states conditioned on the event $\Upsilon''$ satisfy
\begin{align}
    \frac{1}{2}\norm{\Phi_{\text{QKD}} (\bar{\rho}_{X_1^n \Theta_1^n A_1^n})_{|\Omega \wedge \Upsilon''} - \Phi_{\text{QKD}} (\rho'_{X_1^n \Theta_1^n A_1^n})_{|\Upsilon''}}_1 \leq \frac{\epsilon_f^2}{\Pr_{\bar{\rho}}(\Upsilon''| \Omega)}
\end{align}
where $\Pr_{\bar{\rho}}(\Upsilon''| \Omega)$ is the probability for the event $\Upsilon''$ for the state $\Phi_{\text{QKD}} (\bar{\rho}_{X_1^n \Theta_1^n A_1^n| \Omega})$\footnote{We abuse notation while writing the probability this way since the state it is evaluated on is $\Phi_{\text{QKD}} (\bar{\rho}_{X_1^n \Theta_1^n A_1^n | \Omega})$, while we simply use the subscripts $\bar{\rho}$ for $P$. We also write probabilities this way for the state $\Phi_{\text{QKD}} (\rho'_{X_1^n \Theta_1^n A_1^n})$ and $\Phi_{\text{QKD}} ({\sigma}_{X_1^n \Theta_1^n A_1^n})$. This is done for the sake of clarity.}. Using the Fuchs-van de Graaf inequality \cite[Lemma 3.5]{TomamichelBook16}, we can transform this to a purified distance bound
\begin{align}
    P(\Phi_{\text{QKD}} ({\bar{\rho}}_{X_1^n \Theta_1^n A_1^n})_{|\Omega \wedge \Upsilon''}, \Phi_{\text{QKD}} (\rho'_{X_1^n \Theta_1^n A_1^n})_{|\Upsilon''}) \leq \sqrt{\frac{2}{\Pr_{\bar{\rho}}(\Upsilon''| \Omega)}}\epsilon_f.
    \label{eq:dist_bd_after_cond}
\end{align}
Let $d:= D_{\max} (\Phi_{\text{QKD}} (\rho'_{X_1^n \Theta_1^n A_1^n}) || \Phi_{\text{QKD}}(\sigma_{X_1^n \Theta_1^n A_1^n}))$. We have proven an upper bound on $d$ in Eq. \ref{eq:src_corr_tilde_rho_sigma_Dmax_bd}. By definition of $D_{\max}$, we have
\begin{align*}
    \Phi_{\text{QKD}} (\rho'_{X_1^n \Theta_1^n A_1^n}) \leq 2^{d} \Phi_{\text{QKD}}(\sigma_{X_1^n \Theta_1^n A_1^n}).
    \numberthis
\end{align*}
Conditioning both sides on the event $\Upsilon''$ implies that 
\begin{align*}
    \Pr_{\rho'}(\Upsilon'') \Phi_{\text{QKD}} (\rho'_{X_1^n \Theta_1^n A_1^n})_{|\Upsilon''} \leq 2^{d} \Pr_{\sigma}(\Upsilon'') \Phi_{\text{QKD}}(\sigma_{X_1^n \Theta_1^n A_1^n})_{|\Upsilon''}
    \numberthis
\end{align*}
where $\Pr_{\rho'}(\Upsilon'')$ and $\Pr_{\sigma}(\Upsilon'')$ are the probability for $\Upsilon''$ for the states $\Phi_{\text{QKD}} (\rho'_{X_1^n \Theta_1^n A_1^n})$ and $\Phi_{\text{QKD}}(\sigma_{X_1^n \Theta_1^n A_1^n})$ respectively. Therefore, we have 
\begin{align*}
    D_{\max}(\Phi_{\text{QKD}} (\rho'_{X_1^n \Theta_1^n A_1^n})_{|\Upsilon''} || \Phi_{\text{QKD}}(\sigma_{X_1^n \Theta_1^n A_1^n})_{|\Upsilon''}) \leq d + \log \frac{\Pr_{\sigma}(\Upsilon'')}{\Pr_{\rho'}(\Upsilon'')}.
    \numberthis
\end{align*}
Together, with Eq. \ref{eq:dist_bd_after_cond} for $\epsilon_{\text{pa}} := \rndBrk{\frac{2}{\Pr_{\bar{\rho}}(\Upsilon''| \Omega)}}^{\frac{1}{2}} \epsilon_f$, we have that 
\begin{align}
    D^{\epsilon_{\text{pa}}}_{\max} (\Phi_{\text{QKD}} (\bar{\rho}_{X_1^n \Theta_1^n A_1^n})_{|\Omega \wedge \Upsilon''} || \Phi_{\text{QKD}}(\sigma_{X_1^n \Theta_1^n A_1^n})_{|\Upsilon''}) \leq d + \log \frac{\Pr_{\sigma}(\Upsilon'')}{\Pr_{\rho'}(\Upsilon'')}. 
\end{align}
Let $\epsilon_1, \epsilon_2, \epsilon_3 > 0$ be arbitrary parameters. We have
\begin{align*}
    H&^{\epsilon_{\text{pa}} + \epsilon_1 + 2(\epsilon_2 + \epsilon_3)}_{\min} (X_S | E \Theta_1^n \hat{\Theta}_1^n T)_{\Phi_{\text{QKD}} (\bar{\rho})_{|\Omega \wedge \Upsilon''}} \\
    &= H^{\epsilon_{\text{pa}} + \epsilon_1 + 2(\epsilon_2 + \epsilon_3)}_{\min} (\bar{X}_1^n | E \Theta_1^n \hat{\Theta}_1^n T)_{\Phi_{\text{QKD}} (\bar{\rho})_{|\Omega \wedge \Upsilon''}} \\
    &\geq H^{\epsilon_{\text{pa}} + \epsilon_1}_{\min} (\bar{X}_1^n \bar{Y}_1^n | E \Theta_1^n \hat{\Theta}_1^n T)_{\Phi_{\text{QKD}} (\bar{\rho})_{|\Omega \wedge \Upsilon''}} - H^{\epsilon_2}_{\max}(\bar{Y}_1^n | \bar{X}_1^n E \Theta_1^n \hat{\Theta}_1^n T)_{\Phi_{\text{QKD}} (\bar{\rho})_{|\Omega \wedge \Upsilon''}} - 3 g_0(\epsilon_3) \\
    &\geq H^{\epsilon_{\text{pa}} + \epsilon_1}_{\min} (\bar{X}_1^n \bar{Y}_1^n | E \Theta_1^n \hat{\Theta}_1^n)_{\Phi_{\text{QKD}} (\bar{\rho})_{|\Omega \wedge \Upsilon''}} - \log|T| - H^{\epsilon_2}_{\max}(\bar{Y}_1^n | \bar{X}_1^n E \Theta_1^n \hat{\Theta}_1^n T)_{\Phi_{\text{QKD}} (\bar{\rho})_{|\Omega \wedge \Upsilon''}} - 3 g_0(\epsilon_3) \numberthis
    \label{eq:src_corr_Hmin_chain_rule_step}
\end{align*}
where in the first line we have used the fact that given $\Theta_1^n$ and $\hat{\Theta}_1^n$, one can figure out the set $S$ and then $\bar{X}_1^n = X_S (\perp)_{S^c}$ (see Table \ref{tab:var_defn} for definition of the registers), in the second line we have used the chain rule for smooth min-entropy \cite[Theorem 15]{Vitanov13} and in the last line we have used the dimension bound. We have used the chain rule here to reduce our proof to bounding an entropy, which in the perfect source case, can be bound using entropy accumulation \cite[Section 5.1]{Dupuis20}. \\

Now, we can use Lemma \ref{lemm:Hmin_rho_to_Halpha_sigma_using_Dmax} to derive
\begin{align*}
    H^{\epsilon_{\text{pa}} + \epsilon_1}_{\min} &(\bar{X}_1^n \bar{Y}_1^n | E \Theta_1^n \hat{\Theta}_1^n)_{\Phi_{\text{QKD}} (\bar{\rho})_{|\Omega \wedge \Upsilon''}}\\
    &\geq \tilde{H}^{\uparrow}_{\alpha}(\bar{X}_1^n \bar{Y}_1^n | E \Theta_1^n \hat{\Theta}_1^n)_{\Phi_{\text{QKD}}(\sigma)_{|\Upsilon''}} \\
    & \qquad -\frac{\alpha}{\alpha-1} D^{\epsilon_{\text{pa}}}_{\max} (\Phi_{\text{QKD}} (\bar{\rho}_{X_1^n \Theta_1^n A_1^n})_{|\Omega \wedge \Upsilon''} || \Phi_{\text{QKD}}(\sigma_{X_1^n \Theta_1^n A_1^n})_{|\Upsilon''}) - \frac{g_1(\epsilon_1, \epsilon_{\text{pa}})}{\alpha-1} \\
    &\geq \tilde{H}^{\uparrow}_{\alpha}(\bar{X}_1^n \bar{Y}_1^n| E \Theta_1^n \hat{\Theta}_1^n)_{\Phi_{\text{QKD}}(\sigma)_{|\Upsilon''}} \\
    &\qquad - \frac{\alpha}{\alpha-1}d - \frac{\alpha}{\alpha-1} \log \frac{\Pr_{\sigma}(\Upsilon'')}{\Pr_{\rho'}(\Upsilon'')} - \frac{g_1(\epsilon_1, \epsilon_{\text{pa}})}{\alpha-1} \numberthis
    \label{eq:real_st_smooth_min_lower_bd}
\end{align*}
Thus, we have reduced the problem to lower bounding $\alpha$-R\'enyi conditional entropy for the QKD protocol in Protocol \ref{frame:BB84}, where Alice's source produces noisy BB84 states. We can bound this conditional entropy using the entropy accumulation theorem. The only difference in the following arguments from \cite[Section 5.1]{Dupuis20} is that we need to employ entropy accumulation for $\alpha$-R\'enyi entropies (also see \cite{George22}).\\ 

Firstly, note that we can use source purification for the state $\Phi_{\text{QKD}}(\sigma)$, that is, we can imagine that the state $\Phi_{\text{QKD}}(\sigma)$ was produced by the following procedure: 
\begin{enumerate}
    \item Alice prepares $n$ Bell states $(\Phi^{+})^{\otimes n}_{\bar{A}_1^n A_1^n}$.
    \item For each $i \in [n]$, Alice measures the qubit $\bar{A}_i$ in the basis $\Theta_i$, which is chosen to be $Z$ with probability $(1-\mu)$ and otherwise is chosen to be $X$. The measurement result is labelled $X_i$. 
    \item She then applies the $2(\epsilon+ \delta)-$depolarising channel to each of the qubits $A_i$ for $i \in [n]$ and sends them over the channel to Bob.
\end{enumerate}
We can imagine that the source state is prepared in this fashion. The initial state for EAT will be represented by the registers $\bar{A}_1^n A_1^n E$, which contain the state produced after Eve forwards the state produced above by Alice to Bob. We can now define the EAT maps $\mathcal{M}_i: \bar{A}_{i}^n A_{i}^n \rightarrow \bar{A}_{i+1}^n A_{i+1}^n \bar{X}_i \bar{Y}_i \Theta_i \hat{\Theta}_i C_i$, where the registers $\Theta_i$ and $ \hat{\Theta}_i$ are produced by randomly sampling according to the probabilities in the protocol, $\bar{X}_i$ and $\bar{Y}_i$ are produced according to the measurements chosen in the protocol and the source preparation procedure above, and $C_i$ is defined as in Table \ref{tab:var_defn}. \\

Note that by conditioning on the event $\Upsilon''$, we are requiring that $q= \text{freq}(C_1^n)$ satisfies $q(1)\leq e \mu^2$. It is shown in \cite[Proof of Claim 5.2]{Dupuis20} that there exists an affine min-tradeoff function $f$, such that $C_1^n$ given $\Upsilon''$ satisfies $f(\text{freq}(C_1^n)) \geq 1- 2\mu + \mu^2 - h(e)$. Using the entropy accumulation theorem \cite[Proposition 4.5]{Dupuis20}, we get 
\begin{align}
    \tilde{H}^{\uparrow}_{\alpha}(\bar{X}_1^n \bar{Y}_1^n | E \Theta_1^n \hat{\Theta}_1^n)_{\Phi_{\text{QKD}}(\sigma^\delta)_{|\Upsilon''}} \geq n(1- 2 \mu  + \mu^2 - h(e)) - n \frac{\alpha-1}{4}V^2 - \frac{\alpha}{\alpha-1} \log \frac{1}{\Pr_{\sigma}(\Upsilon'')}
    \label{eq:src_corr_Halpha_bd}
\end{align}
where $V:= 2\lceil \norm{\nabla f }_\infty \rceil + 2 \log(1+2|\mathcal{X}|^2) = \frac{2}{\mu^2}\log\frac{1-e}{e} + 2 \log(1+2|\mathcal{X}|^2)$\footnote{It should be noted that this term can be improved using \cite{Dupuis19}.} and $1< \alpha < 1+ \frac{2}{V}$. Combining 
Eq. \ref{eq:real_st_smooth_min_lower_bd} and \ref{eq:src_corr_Halpha_bd}, we get
\begin{align*}
    H^{\epsilon_{\text{pa}} + \epsilon_1}_{\min} &(\bar{X}_1^n \bar{Y}_1^n | E \Theta_1^n \hat{\Theta}_1^n)_{\Phi_{\text{QKD}} (\bar{\rho})_{|\Omega \wedge \Upsilon''}}\\
    & \geq n(1- 2 \mu + \mu^2  - h(e)) - n \frac{\alpha-1}{4}V^2 - \frac{\alpha}{\alpha-1}d - \frac{\alpha}{\alpha-1} \log \frac{1}{\Pr_{\rho'}(\Upsilon'')} - \frac{g_1(\epsilon_1, \epsilon_{\text{pa}})}{\alpha-1} \\
    & \geq n(1- 2 \mu + \mu^2  - h(e)) - n \frac{\alpha-1}{4}V^2- \frac{\alpha}{\alpha-1}n h(\epsilon + \delta) - \frac{\alpha}{\alpha-1} \log \frac{1}{\Pr_{{\rho}} (\Omega) - \epsilon_{\text{qu}}^{\delta}} \\
    &\qquad - \frac{\alpha}{\alpha-1} \log \frac{1}{\Pr_{\bar{\rho}}(\Upsilon''|\Omega) - \frac{2\epsilon^{\delta}_{\text{qu}}}{\Pr_{\rho}(\Omega)}} - \frac{g_1(\epsilon_1, \epsilon_{\text{pa}})}{\alpha-1} \\
    & \geq n(1- 2 \mu + \mu^2  - h(e)) - n \frac{\alpha-1}{4}V^2- \frac{\alpha}{\alpha-1}n h(\epsilon + \delta) \\
    &\qquad - \frac{\alpha}{\alpha-1} \rndBrk{\log \frac{1}{\Pr_{\rho}(\Omega \wedge \Upsilon'') - 2\epsilon^{\delta}_{\text{qu}}} + 1} - \frac{g_1(\epsilon_1, \epsilon_{\text{pa}})}{\alpha-1} \numberthis
\end{align*}
where we have used Eq. \ref{eq:rho_prime_dist_bd}, $\epsilon_f = 2\sqrt{\frac{\epsilon^{\delta}_{\text{qu}}}{\Pr_{\rho}(\Omega)}}$, $\Pr_{\rho}(\Omega \wedge \Upsilon'') = \Pr_{\rho}(\Omega)\Pr_{\bar{\rho}}(\Upsilon'' | \Omega)$ and $\Pr_{\rho}(\Omega) \geq \Pr_{\rho}(\Omega \wedge \Upsilon'') > 2\epsilon^{\delta}_{\text{qu}}$ to simplify the result. It should be noted that the probability $\Pr_{\sigma}(\Upsilon'')$ of the auxiliary state cancels out. Since, we restrict $\epsilon$ and $\delta$ to the region, where $h(\epsilon+\delta)< \frac{1}{\sqrt{2}}$, we can choose
\begin{align}
    \alpha & := 1+ \frac{2\sqrt{2h(\epsilon+\delta)}}{V}
\end{align}
which gives us the bound
\begin{align*}
    H^{\epsilon_{\text{pa}} + \epsilon_1}_{\min} &(\bar{X}_1^n \bar{Y}_1^n | E \Theta_1^n \hat{\Theta}_1^n)_{\Phi_{\text{QKD}} (\bar{\rho})_{|\Omega \wedge \Upsilon''}} \geq n(1- 2 \mu + \mu^2  - h(e) - V \sqrt{2h(\epsilon+\delta)}) \\
    &- \frac{V}{\sqrt{2h(\epsilon+\delta)}} \rndBrk{\log \frac{1}{\Pr_{\rho}(\Omega \wedge \Upsilon'') - 2\epsilon^{\delta}_{\text{qu}}} + 1} - \frac{g_1(\epsilon_1, \epsilon_{\text{pa}})}{2\sqrt{2h(\epsilon+\delta)}} V. \numberthis
    \label{eq:Hmin_AbarA_bd}
\end{align*}
We also need to bound $H^{\epsilon_2}_{\max}(\bar{Y}_1^n | \bar{X}_1^n E \Theta_1^n \hat{\Theta}_1^n T)_{\Phi_{\text{QKD}} (\rho)_{|\Omega \wedge \Upsilon''}}$ in Eq. \ref{eq:src_corr_Hmin_chain_rule_step}. The bound and the proof for this bound are the same as in \cite[Claim 5.2]{Dupuis20}. We have for $\beta \in (1,2)$ that 
\begin{align*}
    H^{\epsilon_2}_{\max}&(\bar{Y}_1^n | \bar{X}_1^n E \Theta_1^n \hat{\Theta}_1^n T)_{\Phi_{\text{QKD}} (\bar{\rho})_{|\Omega \wedge \Upsilon''}} \\
    &\leq H^{\epsilon_2}_{\max}(\bar{Y}_1^n | \Theta_1^n \hat{\Theta}_1^n)_{\Phi_{\text{QKD}} (\bar{\rho})_{|\Omega \wedge \Upsilon''}} \\
    &\leq \tilde{H}^{\downarrow}_{\frac{1}{\beta}}(\bar{Y}_1^n | \Theta_1^n \hat{\Theta}_1^n)_{\Phi_{\text{QKD}} (\bar{\rho})_{|\Omega \wedge \Upsilon''}} + \frac{g_0(\epsilon_2)}{\beta-1}\\ 
    &\leq \tilde{H}^{\downarrow}_{\frac{1}{\beta}}(\bar{Y}_1^n | \Theta_1^n \hat{\Theta}_1^n)_{\Phi_{\text{QKD}} (\bar{\rho})} + \frac{\beta}{\beta-1}\log\frac{1}{\Pr_{\bar{\rho}}(\Omega \wedge \Upsilon'')} + \frac{g_0(\epsilon_2)}{\beta-1} \\ 
    &= \frac{\beta}{\beta-1} \log \sum_{\theta_1^n, \hat{\theta}_1^n} P(\theta_1^n, \hat{\theta}_1^n) 2^{\rndBrk{1-\frac{1}{\beta}}\tilde{H}^{\downarrow}_{\frac{1}{\beta}}(\bar{Y}_1^n | \theta_1^n, \hat{\theta}_1^n)} + \frac{\beta}{\beta-1}\log\frac{1}{\Pr_{\bar{\rho}}(\Omega \wedge \Upsilon'')} + \frac{g_0(\epsilon_2)}{\beta-1}
    \numberthis
\end{align*} 
where the first line follows from the data processing inequality for the smooth max-entropy, second line follows from \cite[Lemma B.10]{Dupuis20}, third line using \cite[Lemma B.6]{Dupuis20}. Let the random variable $Z$ denote the number of $i \in [n]$, such that $\Theta_i = \hat{\Theta}_i =1$. Then, we have the following inequalities for the first term in the bound above
\begin{align*}
    \frac{\beta}{\beta-1} \log \sum_{\theta_1^n, \hat{\theta}_1^n} P(\theta_1^n, \hat{\theta}_1^n) 2^{\rndBrk{1-\frac{1}{\beta}}\tilde{H}^{\downarrow}_{\frac{1}{\beta}}(\bar{A}_1^n | \theta_1^n, \hat{\theta}_1^n)} &\leq \frac{\beta}{\beta-1} \log \sum_{\theta_1^n, \hat{\theta}_1^n} P(\theta_1^n, \hat{\theta}_1^n) 2^{\rndBrk{1-\frac{1}{\beta}}Z_{\theta_1^n, \hat{\theta}_1^n}}\\
    &= \frac{\beta}{\beta-1} \log \sum_{z=0}^n \binom{n}{z} \mu^{2z} (1-\mu^2)^{n-z} 2^{\rndBrk{1-\frac{1}{\beta}}z} \\
    &= n \frac{\beta}{\beta-1}\log\rndBrk{1- \mu^2 + 2^{\rndBrk{1-\frac{1}{\beta}}}\mu^2} \\
    &\leq n\mu^2 \frac{\beta}{(\beta-1)\ln(2)}\rndBrk{2^{\rndBrk{1-\frac{1}{\beta}}}-1} \\
    &\leq  n\mu^2 \rndBrk{1+\frac{(\beta-1)\ln(2)}{\beta}}
    \numberthis
\end{align*}
where we use $Z_{\theta_1^n, \hat{\theta}_1^n}$ to denote the fact that the value of random variable $Z$ is fixed by $\theta_1^n$ and $\hat{\theta}_1^n$, in the second line we transform the expectation over $\theta_1^n$ and $\hat{\theta}_1^n$ into an expectation over $Z$, in the third line we use the binomial theorem, in the fourth line we use the fact that $\ln(1+x) \leq x$ for all $x >-1$, and in the last line we use the fact that $e^x \leq 1 + x + x^2$ for $x\in (0,1)$ and that for $\beta >1$ the term $\ln(2)\rndBrk{1-\frac{1}{\beta}}$ lies in this range. Thus, we get that for $\beta \in (1,2)$, 
\begin{align*}
    H^{\epsilon_2}_{\max}&(\bar{Y}_1^n | \bar{X}_1^n E \Theta_1^n \hat{\Theta}_1^n T)_{\Phi_{\text{QKD}} (\bar{\rho})_{|\Omega \wedge \Upsilon''}} \leq n\mu^2+\frac{(\beta-1)\ln(2)}{\beta} n\mu^2 + \frac{\beta}{\beta-1}\log\frac{1}{\Pr_{\bar{\rho}}(\Omega \wedge \Upsilon'')} + \frac{g_0(\epsilon_2)}{\beta-1}.
    \numberthis
\end{align*}
Choosing $\beta = 1+ \frac{1}{\sqrt{n}}$ and using the coarse bounds $1 < \beta <2$, gives us
\begin{align}
    H^{\epsilon_2}_{\max}&(\bar{Y}_1^n | \bar{X}_1^n E \Theta_1^n \hat{\Theta}_1^n T)_{\Phi_{\text{QKD}} (\bar{\rho})_{|\Omega \wedge \Upsilon''}} \leq n\mu^2+ \sqrt{n} \rndBrk{\mu^2 \ln(2) + 2 \log\frac{1}{\Pr_{\bar{\rho}}(\Omega \wedge \Upsilon'')} +{g_0(\epsilon_2)}}.
    \label{eq:Hmax_barA_bd}
\end{align}
Combining Eq. \ref{eq:src_corr_Hmin_chain_rule_step}, \ref{eq:Hmin_AbarA_bd}, and \ref{eq:Hmax_barA_bd}, we get
\begin{align*}
    H&^{\epsilon_{\text{pa}} + \epsilon_1 + 2(\epsilon_2 + \epsilon_3)}_{\min} (X_S | E \Theta_1^n \hat{\Theta}_1^n T)_{\Phi_{\text{QKD}} (\bar{\rho})_{|\Omega \wedge \Upsilon''}} \\
    &\geq n(1- 2 \mu  - h(e) - V \sqrt{2h(\epsilon+\delta)}) - \sqrt{n} \rndBrk{\mu^2 \ln(2) + 2 \log\frac{1}{\Pr_{\bar{\rho}}(\Omega \wedge \Upsilon'')} +{g_0(\epsilon_2)}}\\
    &- \frac{V}{\sqrt{2h(\epsilon+\delta)}} \rndBrk{\log \frac{1}{\Pr_{\bar{\rho}}(\Omega \wedge \Upsilon'') - 2\epsilon^{\delta}_{\text{qu}}} + 1} - \frac{g_1(\epsilon_1, \epsilon_{\text{pa}})}{2\sqrt{2h(\epsilon+\delta)}} V   - \log|T| - 3 g_0(\epsilon_3)
    \numberthis
    \label{eq:final_Hmin_bd}
\end{align*}
where the parameters $\epsilon_1, \epsilon_2, \epsilon_3 >0$ are arbitrary, and 
\begin{align*}
    \epsilon_{\text{pa}} = 2\rndBrk{\frac{2\epsilon^{\delta}_{\text{qu}}}{\Pr_{\bar{\rho}}(\Omega \wedge \Upsilon'')}}^{1/2}.
    \numberthis
\end{align*} 

For an arbitrary $\epsilon' >0$, we can set $\epsilon_1 = \frac{\epsilon'}{2}$ and $\epsilon_2=\epsilon_3 = \frac{\epsilon'}{8}$ to derive the result in the theorem.

\bibliographystyle{halpha}
\bibliography{bib}

\end{document}